\documentclass[12pt]{article}
\usepackage{amsmath}

\usepackage{amssymb}
\usepackage{amsthm}
\usepackage[all]{xy}
\usepackage{enumerate}
\usepackage{bbm}
\usepackage{bm}
\usepackage{mathtools}
\usepackage{tikz-cd}
\usepackage{graphicx}
\usepackage{hyperref}
\usepackage{accents}
\usepackage{todonotes}

\renewcommand{\indent}{\hspace{1em}}

\theoremstyle{theorem}
\newtheorem{Theorem}{Theorem}[section]
\newtheorem{Proposition}[Theorem]{Proposition}
\newtheorem{Lemma}[Theorem]{Lemma}
\newtheorem{Corollary}[Theorem]{Corollary}
\newtheorem{Propdef}[Theorem]{Proposition and Definition}

\theoremstyle{definition}
\newtheorem{Definition}[Theorem]{Definition}
\newtheorem{Remark}[Theorem]{Remark}
\newtheorem{Example}[Theorem]{Example}

\newcommand{\btm}{\begin{Theorem}}
\newcommand{\etm}{\end{Theorem}}

\newcommand{\ben}{\begin{enumerate}}
\newcommand{\een}{\end{enumerate}}

\newcommand{\bit}{\begin{itemize}}
\newcommand{\eit}{\end{itemize}}

\newcommand{\bca}{\begin{cases}}
\newcommand{\eca}{\end{cases}}

\newcommand{\bre}{\begin{Remark}\rm}
\newcommand{\ere}{\end{Remark}}

\newcommand*{\bbm}{\begin{Remark}}
\newcommand*{\ebm}{\end{Remark}}

\newcommand{\ble}{\begin{Lemma}}
\newcommand{\ele}{\end{Lemma}}

\newcommand*{\bsz}{\begin{Proposition}}
\newcommand*{\esz}{\end{Proposition}}

\newcommand{\beq}{\begin{equation}}
\newcommand{\eeq}{\end{equation}}

\newcommand{\bbma}{\begin{bmatrix}}
\newcommand{\ebma}{\end{bmatrix}}

\newcommand*{\bbs}{\begin{Example}}
\newcommand*{\ebs}{\end{Example}}

\newcommand*{\bfg}{\begin{Corollary}}
\newcommand*{\efg}{\end{Corollary}}

\newcommand*{\bdf}{\begin{Definition}}
\newcommand*{\edf}{\end{Definition}}

\newcommand*{\bbw}{\begin{proof}}
\newcommand*{\ebw}{\end{proof}}

\newcommand*{\bpf}{\begin{proof}}
\newcommand*{\epf}{\end{proof}}

\newtagform{simple}{}{}
\newtagform{standard}{(}{)}

\allowdisplaybreaks

\newcommand{\CC}{{\mathbb{C}}}

\newcommand{\NN}{{\mathbb{N}}}

\newcommand{\RR}{{\mathbb{R}}}
\newcommand{\bbR}{{\mathbb{R}}}

\newcommand{\C}{\mathbb{C}}
\newcommand{\R}{\mathbb{R}}

\newcommand{\N}{\mathbb{N}}

\newcommand{\Gl}{{\rm GL}}
\newcommand{\Real}{{\rm Re}}

\newcommand{\U}{{\mr{U}}}

\DeclareMathOperator{\Ad}{Ad}

\DeclareMathOperator{\tr}{tr}

\newcommand{\mf}[1]{\mathfrak{#1}}
\newcommand{\mr}[1]{\mathrm{#1}}

\newcommand{\comment}[1]{}
\newcommand{\verweis}[1]{}
\renewcommand{\d}{{\mr d}}

\newcommand{\tg}{\mr T}

\newcommand{\g}{\mathfrak{g}}

\DeclareMathOperator{\End}{End}

\newcommand{\calC}{\mathcal{C}}
\newcommand{\calD}{\mathcal{D}}

\newcommand{\calH}{\mathcal{H}}

\newcommand{\calJ}{\mathcal{J}}

\newcommand{\calP}{\mathcal{P}}
\newcommand{\calQ}{\mathcal{Q}}
\newcommand{\calS}{\mathcal{S}}
\newcommand{\calZ}{\mathcal{Z}}

\newcommand{\id}{\operatorname{id}}

\newcommand{\Mat}{{\rm M}}

\newcommand{\hermitian}{\mathsf{H}}
\newcommand{\poshermitian}{\mathsf{H}^+}
\newcommand{\densityops}{\mathfrak{D}}
\newcommand{\imi}{\mr i}

\newcommand{\inprodHS}[1]{\langle #1 \rangle_{\textup{\tiny HS}}}

\numberwithin{equation}{section}

\begin{document}

\title{Stratification of the space of density matrices via invariant theory and Uhlmann's connection}

\author{H.~Jordan, M.J.~Pflaum, G.~Rudolph}

\maketitle

\begin{abstract}
\noindent
The space of $n \times n$ density matrices is the full state space of an
$n$-level quantum system. While it has been  asserted in the literature
on quantum information theory that decomposing this state space into equirank
pieces yields a Whitney
stratification, rigorous proofs have remained absent. Using methods from
invariant theory, this work provides a proof that the rank
decomposition of the matrix algebra constitutes a (b)-regular
stratification in the sense of Whitney. Furthermore, we demonstrate how the
Uhlmann connection - which describes the adiabatic holonomy of mixed states
and reduces to the Berry connection on pure states - can be geometrically treated
as a stratified connection on a stratified vector bundle.
Finally, our appendix provides a self-contained exposition of stratified fiber bundles and the new concept of stratified connections, which may be of independent interest.
\end{abstract}

\newpage

\tableofcontents


%
\section {Introduction}

Consider a quantum system described by a Hamiltonian $\hat H$ being in a stationary (pure) state of $\hat H$. If  $\hat H$ is 
slowly altered along a path in parameter space (describing the environment), then by the adiabatic theorem the system 
will remain in an eigenstate of $\hat H$ for any instant of time. In particular, if $\hat H$ returns to its original form, which corresponds to 
some closed path in parameter space, then the system returns to a state which coincides with the original state up to a phase factor. 
In 1984, Berry  \cite{Berry84} observed that,  in addition to the standard dynamical phase factor, this phase factor contains 
a path-dependent phase, which people used to call Berry's phase. Berry found a general formula for that phase and 
discussed essentially two examples: spins in magnetic fields and the Aharonov-Bohm effect. 
Immediately after the publication of this work, Simon \cite{Simon} clarified the geometric meaning of Berry's  phase. He observed that this phase may be described in terms of the holonomy of a connection in a Hermitian line bundle. 

\indent Over the years, there have been many applications of Berry's phase in various branches of quantum physics,
see \cite{Zwanziger}, \cite{Wang} and the references therein. 

\indent Two years after the appearance of Berry's work, Uhlmann \cite{AU-1} observed that there is a natural 
generalization of Berry's phase to the case of mixed states. In particular, he found a natural connection
\cite{AU-2} whose holonomy governs the adiabatic dynamics along closed paths.
It turns out that Uhlmann's connection is closely related to the concept of purification of densitiy matrices. Thus, let $\mathcal{H}$ 
be the Hilbert space of a quantum system with scalar product $\langle \cdot, \cdot \rangle$. Recall that a density operator $\rho$ is a positive trace class operator with trace one. It defines a mixed state, that is, a positive definite normalized functional on the algebra $\mf A = \mf B (\mathcal{H})$  of bounded operators on
$\mathcal{H}$ by
\[
a \mapsto \rho(a) = \tr(a \rho)  \, .
\]
Applied to a selfadjoint element $a \in \mf A$ or in other words to an observable of the system, the value $\rho (a)$ is interpreted as the expectation value of the observable $a$ for  the system prepared  in the state $\rho$.  
A purification is a mapping which assigns to every mixed state a pure state belonging to some extended Hilbert space $\mathcal{H}^{\rm ext}$  in such a way that expectation values are preserved. The standard choice of a purification is provided by taking
$$
\mathcal{H}^{\rm ext} := \mathcal{H}^\ast \otimes \mathcal{H} \, ,
$$
with the purification $\Phi \in \mathcal{H}^{\rm ext}$ of $\rho$ uniquely determined by the requirement  that
$$
\tr(a \rho) = \langle (\mathbbm 1 \otimes a) \Phi, \Phi \rangle_{\rm ext}  \quad \text{for all }  a \in  \mf A  \ .
$$
Here, $\langle \cdot, \cdot \rangle_{\rm ext}$ is the natural scalar product on $\mathcal{H}^{\rm ext}$ induced from 
$\langle \cdot , \cdot \rangle$. Note that $\mathcal{H}^\ast \otimes \mathcal{H}$ may be naturally identified with the space ${\cal L}^2 (\mathcal{H})$ of Hilbert-Schmidt operators
on $\mathcal{H}$.  Correspondingly, instead of $\langle \cdot, \cdot \rangle_{\rm ext}$ we will 
write $\inprodHS{\cdot, \cdot}$ in that case. 
Moreover, let us denote the space of all density operators on $\mathcal{H}$ by
$\densityops(\mathcal{H})$  or just $\densityops$ and by ${\cal L}^2 (\mathcal{H})_{\neq 0}$ the space of all non-zero Hilbert-Schmidt operators on $\mathcal{H}$. 
Note that in the finite dimensional case, $\densityops$ is topologically a  manifold with boundary.

\indent Uhlmann observed in \cite{AU-1,AU-2} that ${\cal L}^2(\mathcal{H})_{\neq 0}$ is endowed with the structure of a fibered space over $\densityops$, meaning there exists a canonical surjective continuous map
$\hat{\pi}:{\cal L}^2(\mathcal{H})_{\neq 0} \to \densityops$ called the \emph{projection}.
This projection is given by 
\begin{equation}
\label{Fibr-W-Omega}
\hat{\pi}:{\cal L}^2(\mathcal{H})_{\neq 0} \to \densityops\, , \quad a \mapsto \rho:= \frac{a \, a^\ast}{\inprodHS{a,a}} \ .
\end{equation}
Next, Uhlmann found that there is a natural parallel lift operation for $\hat{\pi}$, given by the following (implicitly 
defined) $1$-form $\bf A$  on ${\cal L}^2(\mathcal{H})_{\neq 0}$:
\begin{equation}
\label{UhlmannEqn}
a^\ast \d a -(\d a^\ast) a = a^\ast a {\bf A} + {\bf A} a^\ast a \, .
\end{equation}
Moreover, he observed that on the open and dense submanifold of invertible $a$, $\bf A$ is a connection $1$-form. 
Finally, Uhlmann considered evolution equations on the space of density matrices
\beq
\label{EvEqu-DensM}
 \dot{\rho} = \frac{1}{i \hbar} [\hat H, \rho]\, ,
\eeq
where $\hat H$ is some self-adjoint Hamiltonian of the system. 
In this context, he studied parallel lifts of solutions $t \mapsto \rho_t$ to the purification space
${\cal L}^2(\mathcal{H})_{\neq 0}$. 

Nowadays, the theory of density matrices is a major field of research. We refer the reader to \cite{Bengtsson,BengtssonZyczkowski}
which can serve as a quite exhaustive introductory source. In particular, the reader will find a large number of references therein. One interesting line of research is the study of the geometry of the space  $\densityops_n$ of $(n\times n)$-density matrices, see \cite{Ciaglia,Chruscinski,DittmannUhlmann,D-1,D-2,D-3}. 
Another, closely related subject is the study of the stratification structure, see \cite{Bruening,Budich,Chruscinski,AndreaFranco,DR-1,DR-2}. 
As a rule, the finite-dimensional case is studied, but there are attempts to deal with the infinite-dimensional case as well, see 
\cite{CiagliaJost} and references therein.

In the present paper, we limit our attention to the finite-dimensional case. Then,  $ {\cal H}^\ast \otimes {\cal H}$ may be identified with the space $\Mat_n $ of complex $(n\times n)$-matrices and, thus,  the space of non-normalized density matrices may be identified 
with the subspace $\poshermitian_n \subset \Mat_n$ of all positive semi-definite Hermitian $(n\times n)$-matrices. In this context, 
we are going to present a novel result about the space of $(n\times n)$-density matrices  $\densityops_n$: we will prove that 
$\densityops_n$ is naturally stratified by rank and that this  stratification is {\rm (}b{\rm )}-regular. This result was accomplished by using the following steps:
\begin{enumerate}
\item
Building on earlier results obtained in \cite{DR-1}, we show that the space  $\poshermitian_n$ is naturally homeomorphic to the orbit space of the right $\U(n)$-action on $\Mat_n$. Under this isomorphism, the projection
$\pi_n: \Mat_n \to \poshermitian_n $ corresponds to the orbit map, cf. Propositions \ref{S-Whitney1}, \ref{S-OrbitMap} and 
\ref{S-HSD-Stratification}.

\item
We construct a Hilbert basis for invariant polynomials for the right (or left) action of $\U(n)$ on $\Mat_n$ based on divisibility arguments dating back to Weyl \cite{Weyl.TheClassicalGroups}. As we are dealing here with real algebraic varieties, standard arguments using Hilbert's Nullstellensatz do not work. Using this, together with point 1, one obtains  the above statement for  $\poshermitian_n $, cf. 
Corollary \ref{Corollary1}.

\item
Extending the result from $\poshermitian_n $ to  $\densityops_n$ is then by standard arguments, cf. Corollary \ref{Corollary2}.
\end{enumerate}

We stress that our approach based on the purification stratification substantially differs from considering decompositions of 
the space of density matrices induced by passing to equivalence classes defined by the adjoint group action (which is considered frequently). In short, in our approach, the decomposition by rank is obtained directly according to the natural right $U(n)$-action on the purification space. In contrast, each stratum  of the adjoint action consists of all density matrices that share the same eigenvalues (up to permutation). The stratification so obtained is a refinement of the purification rank stratification: if one collects all the individual adjoint orbits that share the same number of non-zero eigenvalues, their union form one purification rank stratum. Moroever, the strata of the adjoint action do not coincide with the strata defined by rank; in particular, in general they are not submanifolds.  We refer to \cite{AndreaFranco} for a discussion of this issue. 
In this paper, the reader can also find a construction of the Whitney stratification by rank which does not use the purification picture. Moreover, the authors of \cite{AndreaFranco} also discuss the conic structure around the points of each stratum, but they do not prove the (much more intricate) Whitney (b)-regularity.

To complete the picture, in Section \ref{UhlmannConn-finD} we explain Uhlmann's connection in this finite-dimensional setting. 
In particular, we also show how this connection looks like on the non-generic strata. For a whole class of connections containing 
Uhlmann's connection as a special case, we refer to \cite{DR-2}. 

Finally, in Section \ref{Outlook} we comment on the importance und usefulness of the Whitney {\rm b}-regularity and we also make some 
remarks  on the infinite-dimensional case. In the appendix, we provide the reader with a novel stratification theory for 
spaces of density operators.  
Moreover, we include an exposition of the theory of stratified   vector bundles and stratified connections (latter new).

\paragraph{Acknowledgements.} M.J.P.\ acknowledges support by the National Science Foundation under the awards DMS 2055501
and DMS 2406319.  
He would also like to thank the University of Leipzig for its hospitality during his research visits.


\section{The stratification by rank of finite-dimensional matrix algebras}
\label{sec:FiniteDimensionalMatrixAlgebras}

Here, we discuss the case of density matrices of finite dimension. 
In parts, our presentation is along the lines of \cite{DR-1}.

Let  ${\cal H}$ be the Hilbert space  $\C^n$. Then,  ${\cal H}^\ast \otimes {\cal H}$ naturally coincides
with the space $\End(\C^n)$ of linear endomorphisms and we have 
\beq
\label{NatScPr-Mn}
\inprodHS{v, w} = \tr(v^\ast w) \, . 
\eeq
Choosing the standard basis in $\C^n$ (or any other orthonormal basis), one obtains an isomorphism
$$
 {\cal H}^\ast \otimes {\cal H} \cong \Mat_n \, ,
$$
where $\Mat_n$ denotes the space of complex $n\times n$ matrices. In the sequel, this identification will
be used. 
Let us denote by $\poshermitian_n \subset \Mat_n$ the subspace of all positive semi-definite Hermitian $(n\times n)$-matrices which we also
sometimes refer to as the space of \emph{non-normalized density matrices} in $\Mat_n$. 
Note that we include the zero matrix in this space for now.
The normalization 
condition will be imposed at the end of this section. Thus, here, the projection \eqref{Fibr-W-Omega} reads as follows:
\beq
\label{Fibr-M-D}
\pi_n: \Mat_n \to \poshermitian_n \, , \quad a \mapsto \rho:= a a^\ast \, .
\eeq
Note that the graph of this projection is the union of the graphs of the fiber bundles
\beq
\label{Fibr-M-D-strat}
\pi_{n,k} : \Mat_n(k) \to \poshermitian_n(k)\, , \quad a \mapsto \rho:= a a^\ast \ ,
\eeq
where  $\Mat_n(k)$ and $\poshermitian_n(k)$ denote the manifolds of rank $k$ matrices and non-normalized rank $k$ density matrices, respectively. 
In other words, the subspaces of rank $k$ matrices yield subprojections which can be interpreted as fiber bundles.
In the following, we describe this decomposition and the fiber bundles $\pi_{n,k}: \Mat_n(k) \to \poshermitian_n (k)$ in more detail.
Fundamental for this is the right $\U(n)$-action 
\[ \Mat_n\times \U(n)\to \Mat_n ,\quad  (a,u) \mapsto au \ . \]
By right polar decomposition, 
every element $a\in \Mat_n$ can be factorized in the form 
\begin{equation}
   a = |a^*| u \ , 
\end{equation}
where  $|a^*| = \sqrt{aa^*}$ and where $u \in \U(n)$.
The matrix  $|a^*| $ is Hermitian and positive semi-definite and has the same rank $k$ as $a$, so for $a\neq 0$ we have that $|a^*| $ is unitarily equivalent to a diagonal matrix of the  form
\[
  \begin{pmatrix}
  \lambda_1  & 0 & 0 & 0 & 0 & \ldots  & 0 \\
  0 & \lambda_2 & 0 & 0 & 0 & \ldots   & 0 \\
  \vdots &  & \ddots &  & \vdots  & & \vdots \\
  0 & \ldots & 0 & \lambda_k & 0 &\ldots & 0 \\
  0 & \ldots  & 0 & 0 & 0 & \ldots & 0 \\
  \vdots & & \vdots & \vdots & \vdots & \ddots & \vdots \\
  0 & \ldots & 0  & 0 & 0 & \ldots & 0
  \end{pmatrix}
\]
where $\lambda_1 \geq \lambda_2 \geq \ldots \geq \lambda_k >0$.  The isotropy group of such a diagonal matrix under the $\U(n)$-action 
consists of all unitaries leaving the top left $(k\times k)$-block, the bottom right
$((n-k)\times (n-k))$-block of zeros and the upper right $k \times (n-k))$-block of zeros invariant and
therefore coincides with the image of $\U(n-k)$ in $\U(n)$ under the embedding 
$u \mapsto \begin{pmatrix} \mathbbm{1}_{k} & 0 \\ 0 & u \end{pmatrix}$. 
One concludes that the set $\Mat_n (k)$ of rank $k$ matrices coincides with the set $\Mat_{(U(n-k))} $ of all
elements of $\Mat_n$ whose isotropy groups are conjugate to $\U(n-k)$. By the theory of transformation groups, the latter set is a submanifold of the $\U(n)$-manifold $\Mat_n$; see e.g.~\cite{Bredon}, \cite{JaenichDM} or \cite[Cor.~4.2.8]{PflaumBook}. 
Using right polar decomposition and the diagonalizability of Hermitian matrices again one concludes that
\begin{equation}
  \label{eq:matrix_rank_condition_frontiers}
  \Mat_n (l) \subset \overline{\Mat_n(k)} \quad  \text{for } l \leq k \ 
\end{equation}
hence the decomposition of $\Mat_n$ into submanifolds of matrices of a fixed rank fulfills the condition of frontier and thus is a stratification.   
Since the stratification by rank coincides with the stratification by orbit type,  it fulfills Whitney's condition
(b) by \cite[Thm.~4.3.7]{PflaumBook}. To summarize, we have the following.

\bsz
\label{S-Whitney1}
The decomposition of $\Mat_n$ into submanifolds by rank is a stratification and $(b)$-regular in the sense of Whitney.  
\esz

%
The following crucial observation allows to apply the theory of transformation groups in our setup.
\bsz
\label{S-OrbitMap}
The space  $\poshermitian_n$ of positive semi-definite Hermitian matrices is naturally homeomorphic to 
the orbit space of the right $\U(n)$-action on $\Mat_n$. Under this isomorphism, the projection
$\pi_n: \Mat_n \to \poshermitian_n $ corresponds to the orbit map. 
\esz
\bpf
Denote by $X_n = \Mat_n/\U(n)$ the orbit space and by
$\sigma_n: \Mat_n \to X_n$ the canonical projection.  Since the map $\pi_n$ is invariant under the right $\U(n)$-action, one obtains a commutative diagram
\[
\begin{tikzcd}
\Mat_n \arrow[r,"\sigma_n"] \arrow[rd,"\pi_n",swap]& X_n \arrow[d,"\varphi_n"]\\
& \poshermitian_n 
\end{tikzcd}
\]
By right polar decomposition, the map $\poshermitian_n \to X_n, \: \rho \mapsto \sigma_n (\sqrt{\rho})$ then is a continuous inverse to the map $\varphi_n: X_n \to \poshermitian_n$. 
\epf

Let us  come back to the decomposition of the projection
$\pi_n:\Mat_n\to \poshermitian_n$ into the  bundles
$\pi_{n,k}: \Mat_n(k) \to \poshermitian_n(k) $. 
One concludes from
\eqref{eq:matrix_rank_condition_frontiers} and continuity of $\pi_n$ that
\begin{equation}
\label{eq:density_matrix_rank_condition_frontiers}
  \poshermitian_n (l) \subset \overline{\poshermitian_n(k)} \quad \text{for } l\leq k 
\end{equation}
which implies that the condition of frontier
holds for the decomposition of $\pi_n$ into the bundles $\pi_{n,k}$, hence it is a stratification, indeed.
The regular, or in other words the open and dense stratum is given by
$\pi_{n,n}: \Gl (n,\CC)  \to \poshermitian_n(n) $. The projections $\pi_{n,k}$ 
with $k < n$ constitute the singular strata.
To describe the structure of the bundles $\pi_{n,k}$, recall from our discussion before that $\Mat_n(k)$ consists of all points of $\Mat_n$ having isotropy type $(\U(n-k))$. By a standard theorem
about proper Lie group actions with one orbit type, see e.g.\ \cite[Thm.\  5.8,  Chpt.\ II]{Bredon}, \cite[Thm.\ 4.3.10]{PflaumBook} or 
\cite[Sec.\ 6.6]{RS-1}, 
the projection $\pi_{n,k}: \Mat_n(k) \to \poshermitian_n(k)$ 
is locally trivial and constitutes a bundle with typical fiber $ \U(n-k)\backslash \U(n)$  and structure group
\[
\Gamma_{\U(n-k)} := N_{\U(n)} (\U(n-k))/\U(n-k)\, ,
\]
where $N_{\U(n)} (\U(n-k))$ denotes the normalizer of $\U(n-k)$ in $\U(n)$.
The bundle $\pi_{n,k}$ is associated with the $\Gamma_{\U(n-k)}$-principal bundle 
\[
\rho_{n,k}: \Mat_{\U(n-k)} \to \poshermitian_n(k), \quad a \mapsto a a^*\ , 
\]
where $\Mat_{\U(n-k)} \subset \Mat_n(k)$ denotes the set of all $(n \times n)$-matrices having isotropy
group $\U(n-k)$. Note that  $\Mat_{\U(n-k)}$ is a proper submanifold of $\Mat_n(k)$. 
One easily checks that for the stabilizer $\U(n-k)$ one has a natural isomorphism
\[
\Gamma_{\U(n-k)} \cong \U(k) 
\]
under the embedding of $\U(k)$ in the upper $(k \times k)$-diagonal block of $\U(n)$.
These observations together  with Proposition \ref{S-OrbitMap} entail the following.
\bsz
\label{S-HSD-Stratification}
In the above notation, we have the following commutative diagram:
\[
\begin{tikzcd}
  \Mat_{\U(n-k)} \times_{\U(k)} \U(n-k) \backslash \U(n) \arrow[rr,"\Phi_{n,k}"]
  \arrow[d,"\pi_{\U(n)/\U(n-k)}",swap] 
  & & \Mat_n (k) \arrow[d,"\pi_{n,k}"] 
\\
X_n \supset \Mat_{\U(n-k)}/ \U(k) \ar[rr,"\varphi_{n,k}"]
 & &
\poshermitian_n(k) \, ,
\end{tikzcd}
\]
where $\varphi_{n,k} := \varphi_n|_{\Mat_{\U(n-k)}/ \U(k)}$ is a diffeomorphism
and $\Phi_{n,k}$ is the bundle isomorphism defined by 
\beq
\label{Iso-Phi}
\Phi ([(a, [u])]) := a u \quad \text{for } a\in  \Mat_{\U(n-k)} \text{ and } u \in \U(n) \, . 
\eeq
\esz

For later purposes, we denote 
\[
E_n(k) :=\Mat_{\U(n-k)} \times_{\U(k)} \U(n-k) \backslash \U(n)\, .
\]

\begin{Remark}
\label{B-GeomStrata}
In the following we discuss the geometry of the strata
in more detail.
\begin{enumerate}[1.]
\item\label{Ite1-B-GeomStrata}
By the above discussion, the elements of  the subset of isotropy type $\Mat_{\U(n-k)}$ are built from $(n \times k)$-blocks, whose colums are linearly independent vectors in $\C^n$, and of $(n \times (n-k))$-blocks whose entries vanish. Thus, $\Mat_{\U(n-k)}$ may be identified with the space of linear $k$-frames $P_n(k) $ in $\C^n$, because any $k$-frame $f$ may be regarded as an injective mapping $f: \C^k \to \C^n$ and thus as an $(n \times k)$-matrix of rank $k$. The space $P_n(k) $ is usually referred to as the non-compact Stiefel 
manifold of type $(n,k)$. Thus, the $\U(k)$-principal bundle $\rho_{n,k}$ may be identified with the 
$\U(k)$-principal bundle of linear $k$-frames over $\poshermitian_n(k)$, 
\[
\lambda_{n,k}: P_n (k) \to \poshermitian_n(k) \, .
\]
Next, let $p_{n,k}: \C^n \to \C^k$ be the canonical projection onto the first $k$ coordinates.
It may be regarded as a $(k \times n)$-matrix fulfilling $p_{n,k} p_{n,k}^\ast = \id_{\C^k}$.
Under the above identifications, the bundle isomorphism $\Phi_{n,k}$ then turns into the mapping
\beq\label{Struct-M-n-k}
  \Phi_{n,k} : P_n(k) \times_{\U(k)}  \U(n-k) \backslash \U(n)  \to \Mat_n (k)\, ,
  \quad [(f,[u])] \mapsto f p_{n,k} u \, ,
\eeq
and the canonical projection $\lambda_{n,k}$ reads
$\lambda_{n,k} (f) = f f^\ast$. Since  $p_{n,k} p_{n,k}^\ast = \id_{\C^k}$, one has  
$$ 
f p_{n,k} u (f p_{n,k} u)^\ast = f f^\ast 
$$ 
showing that $\Phi_{n,k}$ maps fibers into  fibers, indeed. It is also obvious that the restriction of
$\Phi_{n,k}$ to any fiber is a bijection. 
\item\label{Ite2-B-GeomStrata}
  Recall the compact Stiefel bundle
  \[
    \overline{\lambda}_{n,k}: S_n(k) \to G_n(k)\, ,
    f \mapsto f f^* \, ,
  \]
  where $S_n(k)$ and $G_n(k)$ denote the spaces of orthonormal $k$-frames and of rank $k$
  orthogonal projections of $\C^n$, respectively. Since an element $ f \in S_n(k)$ can be regarded
  as an injective isometric map $f :\C^k \to \C^n$ one concludes that the equality 
\beq
\label{Emb-Sn-Pn}
  f^\ast f = \id_{\C^k} 
\eeq
  holds true. Moreover for $0<k<n$ the Stiefel  bundle $\overline{\lambda}_{n,k}$ is non-trivial
  and, even for $k=0,n$, the bundle $\overline{\lambda}_{n,k}$  is  a subbundle of 
  $\lambda_{n,k}: P_n(k) \to \poshermitian_n(k) $ via the identical embeddings 
  of $S_n(k)$ into $P_n(k)$ and of $G_n(k)$ into $\poshermitian_n (k)$.
On the other hand, the mapping $i_{n,k}: P_n(k) \to \Mat_n(k) $ defined by 
\beq
\label{Emb-Pn-Mn}
i_{n,k} (f) := fp_{n,k} \quad \text{for } f \in P_n(k) 
\eeq
is an injective bundle homomorphism. Thus, the compact Stiefel bundle $\overline{\lambda}_{n,k} $ is 
a subbundle of $\pi_{n,k}$. Since $S_n(k)$ is a deformation retract of $\Mat_n(k)$, the bundle
$\pi_{n,k}$ therefore is non-trivial whenever $\overline{\lambda}_{n,k}$ is.
\item\label{Ite3-B-GeomStrata}
  As pointed out before, the regular (dense) stratum is given by the bundle
  \beq
  \label{Fibr-M-D-princStr}
  \pi_{n,n}: \Gl(n, \C) \to  \poshermitian_n(n) \, .
  \eeq
  For every element $a \in \Gl(n,\C)$ the unitary $u$ in the right polar decomposition $a= |a^*|u$ 
  is uniquely determined, hence $\pi_{n,n}$ is a principal fiber bundle with structure group $\U(n)$. 
  Since it admits the global section 
  $\poshermitian_n (n) \to \Gl(n, \C) $, $\rho \mapsto \sqrt{\rho}$, the bundle
  $\pi_{n,n}$ is trivial. Furthermore, the base space $\poshermitian_n (n)$ can be identified with
  the quotient space $\Gl(n,\C)/ \U(n)$ via the natural isomorphism
  \[
           \poshermitian_n (n) \to \Gl(n,\C)/ \U(n), \quad \rho \mapsto \sqrt{\rho} \, \U(n)   \, .
  \] 
\end{enumerate}
\end{Remark}

%


\section{Invariant theory of the right $\U (n)$-action}
\label{Space-densMatr}


To start with, recall from the previous section that the projections $\pi_{n,k}: \Mat_n(k) \to \poshermitian_n(k)$ are fiber bundles and that the stratification of $\Mat_n$ by rank induces the decomposition of $ \poshermitian_n$ into pieces $\poshermitian_n(k)$ of constant rank via the projection $\pi_n:\Mat_n \to \poshermitian_n$.
By \eqref{eq:density_matrix_rank_condition_frontiers}, this decomposition
fulfills the condition of frontier, so is a stratification indeed.
Below, we are going to prove that the stratification of $ \poshermitian_n$ into the
manifolds  $\poshermitian_n(k)$ is actually (b)-regular in the sense of Whitney.
To this end we will apply methods from semialgebraic geometry and invariant theory.
Our first observation is the following result.

\begin{Proposition}
  The spaces $\poshermitian_n$ and $\densityops_n$ of positive semi-definite Hermitian $n\times n$
  matrices and $n\times n$ density matrices, respectivcely, are semialgebraic subsets of $\Mat_n$.
\end{Proposition}

\bpf
Obviously, the space $\mathsf{H}_n$ of Hermitian $n\times n$ matrices is a real vector space.
Denoting by $\operatorname{diag} (\lambda_1,\ldots ,\lambda_n)$ the diagonal matrix
with entries $\lambda_1,\ldots,\lambda_n \in \R$, the set
\[
  \left\{  (a,u,\lambda) \in \mathsf{H}_n\times \Mat_n \times \RR^n \mid
    u \, u^* = \mathbbm{1} \ , \: u a u^* = \operatorname{diag} (\lambda_1,\ldots ,\lambda_n) \ ,
    \: \lambda_1 \geq \ldots \geq \lambda_n  \geq 0  \right\}
\]
then is a semialgebraic subset of $\mathsf{H}_n\times \RR$. The projection 
of this set to  $\mathsf{H}_n$ coincides by definition with the space of positive semi-definite Hermitian matrices, so by Tarski--Seidenberg it is semialgebraic.

\indent
The space of $n\times n$ density matrices is the intersection of two semialgebraic sets, namely of
$\poshermitian_n$ and the linear subspace $V= \{ a \in \Mat_n \mid \tr (a \, a^*) =1  \}$,
so is semialgebraic as well. 
\epf

As a semialgebraic set, $\poshermitian_n$ possesses a minimal Whitney stratification by algebraic
manifolds. We will prove in the following that the minimal Whitney stratification is given by the
stratification by rank. 

For the convenience of the reader we recall some basic notions of invariant theory. 
Let $G$ be a compact (or reductive) Lie group and assume that $G$ acts by linear transformations from 
the right on the vector space $\RR^d$. The ring of complex valued polynomials on $\RR^d$ invariant 
under the $G$-action will be denoted by $\C[x_1,\ldots,x_d]^G$. 
By the Hilbert basis theorem, this ring is finitely generated, that is,
there exist invariant complex valued polynomials $p_1,\ldots,p_k$ on $\RR^d$ such that    
\[
  \C[x_1,\ldots,x_d]^G = \C [p_1,\ldots,p_k] \ .
\]
The map 
\[
  p: \RR^d \to \CC^k \ , \quad x \mapsto 
  \begin{pmatrix}
    p_1 (x) \\ \vdots \\ p_k (x)
  \end{pmatrix}
\]
then is referred to as the \emph{Hilbert map} of the $G$-action.
The claim that $p_1,\ldots,p_k$ form a Hilbert basis of invariant polynomials can then be briefly written in the form
\[
   p^* (\C[z_1, \ldots ,z_k] )= \C[x_1, \ldots , x_d]^G \ .
\]
A Hilbert basis and the corresponding Hilbert map are called \emph{minimal} if the Hilbert basis is of minimal length.
The crucial  observation now is that given a Hilbert map $p: \RR^d \to \CC^k$,
the image $X = p(\RR^d)$ can be identified with the orbit space of the
$G$-action and that $X$ carries a natural minimal stratification
whose strata coincide with the images of the orbit type strata of $\RR^d$
under the Hilbert map. Moreover, this stratification is {\rm (}b{\rm )}-regular in the sense of Whitney. It is called the \emph{stratification} of $X$ by \emph{orbit type}.
See Examples \ref{ex:differentiable-stratified-space-structure-orbit-space} and \ref{ex:hilbert.map-whitney-stratification}
in the appendix or \cite{BierstoneOrbitspaces} and \cite[Section 4.4]{PflaumBook} for further details.
The following result is fundamental.

\begin{Theorem}
  \label{thm:HilbertMap}
  The map
  \[
    \pi_n : \Mat_n \to   \poshermitian_n \subset \Mat_n \ , \quad a \mapsto a a^*
  \]
  is a minimal Hilbert map for the right $\U(n)$-action on $\Mat_n$.
\end{Theorem}

Before we can prove the theorem we need the following crucial tool.

\begin{Lemma}[{\cite[Lemma 1-1.7.]{KacInvariantTheory}}] \label{thm:kac}
Let $G$ be a group acting (from the left or the right) on a real  or complex algebraic variety $X$ and let $p_1,\ldots,p_k$ be invariant polynomials. Suppose $H\subset G$ is a subgroup which leaves a subvariety $Y$ of  $X$ invariant such that
\begin{enumerate}[\rm (1)]
\item\label{ite1}
  The functions $p_i|_Y$ generate the ring of all $H$-invariant polynomials on $Y$.  
\item\label{ite2}
  The set $GY$ is a dense subset of $X$ in the euclidean topology. 
\end{enumerate}
Then   the polynomials $p_1,\ldots , p_k$ generate the ring of $G$-invariant polynomials on $X$.
%
\end{Lemma}

\bpf
  Let $p$ be a polynomial on $X$ invariant under the $G$-action. Then the restricted polynomial $p|_Y$ 
  is invariant under the $H$-action. By assumption (\ref{ite1}) there exists a polynomial $q$ in $k$ variables such that 
  \[
      f|_Y = q \big (p_1|_Y, \ldots , p_k|_Y \big) \  .
  \]
  By $G$-invariance of the left and the right side one concludes that
  \[
   f|_{GY} = q \big (p_1|_{GY}, \ldots , p_k|_{GY} \big)  \  .
  \]
  The density assumption (\ref{ite2}) then implies
  $ f = q (p_1,\ldots , p_k)$,
  which proves the claim.
\epf

\begin{proof}[Proof of Theorem \ref{thm:HilbertMap}]
  We prove the claim by induction on $n$. For $n=1$, the algebra of $\U(1)$-invariant
  polynomials on $\C \cong \R^2$ is generated by the map
  $\pi_1 : \C \to \R $, $z \mapsto |z|^2$. See e.g.\ \cite[Prop.~3.1]{CraigWattsEtAl} for a proof. 
  Obviously, the Hilbert basis of invariant polynomials is minimal in this case.

  Now assume the claim holds for some natural $n >0$. Consider the variety $Y\subset \Mat_{n+1}$ consisting
  of all matrices $a \in \Mat_{n+1}$ of the form
  \[
    a = \begin{pmatrix} \widetilde{a} & v \\ 0 & v_{n+1} \end{pmatrix} \ , 
  \]
  where $\widetilde{a} \in \Mat_{n}$, $v = \begin{pmatrix} v_1 \\ \vdots \\ v_{n} \end{pmatrix} \in \CC^{n} $, and $v_{n+1} \in \C^*$.
  Furthermore, let $H \subset \U(n+1)$ be the group of all unitary matrices  $u$ of the form
  \[
     u = \begin{pmatrix} \widetilde{u} & 0 \\ 0 & z \end{pmatrix} \ , 
  \]
  where $\widetilde{u} \in U(n)$ and $z\in \U(1)$.
  Let 
  $ 
      f \in \C\otimes \R [\Mat_{n+1}]\cong \C
      [ (z_{ij} )_{{\mbox{}}_{1\leq i,j \leq n+1}}, \,
      (\overline{z}_{ij})_{{\mbox{}}_{1\leq i,j \leq n+1}} ]  
  $
  be a polynomial invariant under the right $\U(n+1)$-action on $\Mat_{n+1}$ and homogeneous of degree $d$ . Then $f$ can be written in the form
  \begin{displaymath}
  \begin{split}
      f & = \\ 
      & \sum_{ \substack{k_0,k_1,\overline{k}_1,\ldots , k_{n+1},\overline{k}_{n+1}\in \NN \\
      k_0 + k_1 + \overline{k}_1 + \dots + k_{n+1} + \overline{k}_{n+1}=d 
      }} \hspace{-10mm}
      f_{k_0} \big( (z_{ij} )_{{\mbox{}}_{\substack{\scriptscriptstyle 1\leq i \leq n+1 \\ \scriptscriptstyle 1\leq j \leq n}}} , \,
      (\overline{z}_{ij})_{{\mbox{}}_{\substack{\scriptscriptstyle 1\leq i \leq n+1 \\ \scriptscriptstyle 1\leq j \leq n}}} \big) 
      z_{1(n+1)}^{k_1}  \, \overline{z}_{1(n+1)}^{\overline{k}_1}   \cdot \ldots \cdot  z_{(n+1)(n+1)}^{k_{n+1}}  \, \overline{z}_{(n+1)(n+1)}^{\overline{k}_{n+1}}    \  ,
 \end{split}
 \end{displaymath}
 where $f_{k_0} $ is a polynomial homogeneous of degree $k_0$ in the variables $z_{ij}$ and $\overline{z}_{ij}$ with $1\leq i \leq n+1$
 and $1\leq j \leq n$. 
 After restriction of $f$ to the subspace $Y \subset \Mat_{n+1}$, the resulting polynomial is invariant under the right action of $H$ and can be written as 
 \begin{displaymath}
  \begin{split}
      f |_Y& = \\ 
      & \sum_{ \substack{k_0,k_1,\overline{k}_1,\ldots , k_{n+1},\overline{k}_{n+1}\in \NN \\
      k_0 + k_1 + \overline{k}_1 + \dots + k_{n+1} + \overline{k}_{n+1}=d }} \hspace{-10mm}
      \widetilde{f}_{k_0} \big( (z_{ij} )_{{\mbox{}}_{1\leq i,j \leq  n}} , \,
      (\overline{z}_{ij})_{{\mbox{}}_{ 1\leq i,j \leq n }} \big) 
      z_{1(n+1)}^{k_1}  \, \overline{z}_{1(n+1)}^{\overline{k}_1}   \cdot \ldots \cdot  z_{(n+1)(n+1)}^{k_{n+1}}  \, \overline{z}_{(n+1)(n+1)}^{\overline{k}_{n+1}}    \  ,
 \end{split}
 \end{displaymath}
 where $\widetilde{f}_{k_0}$ is the polynomial obtained from $f_{k_0}$ by setting 
 \[ z_{(n+1)1} = \ldots =  z_{(n+1)n} = 0 \text{ and } \overline{z}_{(n+1)1} = \ldots =  \overline{z}_{(n+1)n} = 0  \ . \]
 Note that for each $k_0 \leq d$ the polynomial $\widetilde{f}_0$ depends on the variables $z_{ij}$ and $\overline{z}_{ij}$ with $1\leq i,j \leq n$ only,
 is homogeneous   of degree $k_0$ and uniquely determined by $f|_Y$. Moreover, the monomial  
 $ z_{1(n+1)}^{k_1}  \, \overline{z}_{1(n+1)}^{k_1}   \cdot \ldots \cdot  z_{(n+1)(n+1)}^{k_{n+1}}  \, \overline{z}_{(n+1)(n+1)}^{k_{n+1}}   $ depends only on  the last   
 column of a matrix in $Y$ and is homogeneous of degree $d-k_0= k_1+\overline{k}_1+ \ldots + k_{n+1}+\overline{k}_{n+1}$. By invariance of $f|_Y$ under the right $H$-action on $Y$ and since $H \cong \U(n) \times \U(1)$  one concludes by homogeneity that $\widetilde{f}_{k_0}$ is $\U(n)$-invariant and that
 the monomial 
 \[
   m_{k_1,\overline{k}_1,\ldots, k_{n+1},\overline{k}_{n+1}} = 
   z_{1(n+1)}^{k_1}  \, \overline{z}_{1(n+1)}^{\overline{k}_1} \cdot \ldots \cdot  z_{(n+1)(n+1)}^{k_{n+1}}  \, 
   \overline{z}_{(n+1)(n+1)}^{\overline{k}_{n+1}}   
 \]
 is $\U(1)$-invariant whenever $\widetilde{f}_{k_0} \neq 0$. 
 By the inductive hypothesis, the polynomial $\widetilde{f}_0$ can now be represented in the form 
 \[
    \widetilde{f}_0 = \widetilde{F}_0 \big( (\pi_n^{ij})_{1\leq i,j \leq n}  \big) = 
    \widetilde{F}_0 \big( ( \langle z_n^i,z_n^j \rangle)_{1\leq i,j \leq n}  \big) \ ,  
 \]  
 where  $ \widetilde{F}_0$ is a polynomial in $n^2$ variables and 
 \[
   \pi_n^{ij} ( z )  = \langle z_n^i, z_n^j \rangle = \sum_{l=1}^n z_{il}\overline{z}_{jl}  \quad 
   \text{for } z = (z_{ij})_{1\leq i,j \leq n} \in \Mat_n \ .
 \] 
 Note that here we have used the symbol $z_n^i$ to denote the row vector $(z_{i1}, \ldots , z_{in})$.
 Since the monomial $m_{k_1,\overline{k}_1,\ldots, k_{n+1},\overline{k}_{n+1}}$
 is $\U(1)$-invariant, the equality
 \begin{equation}
 \label{eq:condition-invariance-monomials}
   \sum_{i=1}^{n+1} k_i  = \sum_{i=1}^{n+1} \overline{k}_i
  \end{equation}
 holds by \cite[Prop.\ 3.1]{CraigWattsEtAl}.  Hence, the monomial 
 $ m_{k_1,\overline{k}_1,\ldots, k_{n+1},\overline{k}_{n+1}} $ has a representation of the  form
 \[
   m_{k_1,\overline{k}_1,\ldots, k_{n+1},\overline{k}_{n+1}} = 
   |z_{1(n+1)}|^{2l_{11}} \cdot \ldots \cdot |z_{(n+1)(n+1)}|^{2l_{(n+1)(n+1)}} \cdot 
   \prod_{\substack{1 \leq i,j\leq n+1 \\ i\neq j}} (z_i \overline{z}_j)^{l_{ij}} \ ,
 \] 
 where the $l_{ij}$ are natural numbers and $\sum_{1\leq i,j\leq n+1} l_{ij} = \sum_{1\leq i \leq n+1} k_i$. Note that the 
 $l_{ij}$ are not uniquely determined, but that one can choose $l_{ii}= \min \{ k_i, \overline{k}_i\}$.
 In the next step we express the polynomials $\pi^{ij}_n$ and the monomials $|z_{(1(n+1)}|^2, \ldots , |z_{(n+1)(n+1)}|^2$
 over the subspace $Y \subset \Mat_{n+1}$ in terms of the polynomials 
 \[
   \pi_{n+1}^{ij} (z) = \langle z_{n+1}^i , z_{n+1}^j \rangle = \sum_{l=1}^{n+1} z_{il} \overline{z}_{jl} \ ,
 \]
 where  $1 \leq i,j \leq n+1$, $ z =(z_{ij})_{1\leq i,j\leq n+1} \in \Mat_{n+1}$, and $z_{n+1}^i = (z_{i1} , \ldots , z_{i(n+1)})$.   
 For $ z =(z_{ij})_{1\leq i,j\leq n+1} \in Y$ one obtains
 \[
   |z_{(n+1)(n+1)}|^2  = z_{(n+1)(n+1)}\overline{z}_{(n+1)(n+1)} = \langle z_{n+1}^{n+1}, z_{n+1}^{n+1} \rangle =
   \pi_{n+1}^{(n+1)(n+1)} (z) 
 \] 
 and, under the condition that  $1 \leq i,j \leq n$ and $i\neq j$,
 \begin{displaymath}
 \begin{split}
    |z_{i(n+1)}|^2 & = \langle z_{n+1}^i, z_{n+1}^{n+1} \rangle =  \pi_{n+1}^{i(n+1)} (z) \ , \\
    z_{i(n+1)} \overline{z}_{j(n+1)} & = 
    \frac{1}{\langle z_{n+1}^{n+1}, z_{n+1}^{n+1} \rangle} \cdot \langle z_{n+1}^i, z_{n+1}^{n+1} \rangle \cdot 
    \langle z_{n+1}^{n+1}, z_{n+1}^j \rangle  =    \\
    & =    \frac{1}{\pi_{n+1}^{(n+1)(n+1)} (z)} \cdot \pi_{n+1}^{i(n+1)} (z)  \cdot \pi_{n+1}^{(n+1)j} (z) \  .
 \end{split}
 \end{displaymath}
 In addition, the following equality holds for $1\leq i,j \leq n$, again under the assumption $ z \in Y$:
 \begin{displaymath}
   \pi_n^{ij} (z) =  \pi_{n+1}^{ij} (z)  - z_{i(n+1)} \overline{z}_{j(n+1)} = \pi_{n+1}^{ij} (z)  - 
   \frac{1}{\pi_{n+1}^{(n+1)(n+1)} (z)} \cdot \pi_{n+1}^{i(n+1)} (z)  \cdot \pi_{n+1}^{(n+1)j} (z) \  .
 \end{displaymath}
 Hence the $H$-invariant polynomials $\pi_n^{ij}$ and the $H$-invariant monomials  
 $m_{k_1,\overline{k}_1,\ldots, k_{n+1},\overline{k}_{n+1}}$ can be represented over 
 $Y$ as polynomials in the functions $\pi_{n+1}^{i(n+1)}\big|_Y$ and $\overline{\pi}_{n+1}^{i(n+1)}\big|_Y$ 
 with $1\leq i \leq n$,  $\pi_{n+1}^{(n+1)(n+1)}\big|_Y$ and $\frac{1}{\pi_{n+1}^{(n+1)(n+1)} \big|_Y}$.  
 By the proof of Lemma \ref{thm:kac} one concludes that over the subspace $Y\U(n+1) \subset \Mat_{n+1}$, the 
 $\U(n+1)$-invariant polynomial $f$ can be expressed as
  \begin{equation}\label{eq:invariant-polynomial-representation-upper-triangular}
   f =  \frac{1}{\big(\pi_{n+1}^{(n+1)(n+1)}\big)^r} \, 
   p \big( (\pi_{n+1}^{ij})_{1\leq i,j\leq n+1} )\big)  \ ,
  \end{equation}
  where $r\in \NN$ and $p$ is a polynomial in $(n+1)^2$ real variables.
  
  In an analogous construction, we let $Y^\prime \subset \Mat_{n+1}$ be the subspace of all matrices
  $a \in \Mat_{n+1}$ of the form 
  \[
     a = 
     \begin{pmatrix}
        v_1 & 0 \\ v & \widetilde{a}
     \end{pmatrix} \ ,
  \]
  where $v_1\in \C^*$, $v= \begin{pmatrix} v_2 \\ \vdots  \\ v_{n+1} \end{pmatrix}\in \CC^n$, and $\widetilde{a} \in \Mat_n$.
  The subgroup $H^\prime \cong \U(1) \times \U (n)\subset  \U (n+1)$ is defined as the group of all unitaries $u$ of the form
  \[
     u = \begin{pmatrix}
       z & 0\\
       0 & \widetilde{u}
     \end{pmatrix}
  \]
  where $z \in \U(1)$ and $\widetilde{u}\in \U(n)$. Then $Y^\prime$ is invariant under the right action of $H^\prime$ 
  and we can essentially repeat the argument above just with $Y$ replaced by $Y^\prime$ and $H$ by $H^\prime$.  
  Thus one shows that over the subset $Y^\prime \U(n+1)$ the polynomial $f$ also has a representation of the form 
   \begin{equation}\label{eq:invariant-polynomial-representation-lower-triangular}
   f =  \frac{1}{\big(\pi_{n+1}^{11}\big)^s} \, 
   q \big( (\pi_{n+1}^{ij})_{1\leq i,j\leq n+1} )\big)  \ ,
  \end{equation}
  where $s\in \NN$ and $q$ is a polynomial in $(n+1)^2$ real variables.  
  For later purposes note that the polynomial $p$ (and likewise $q$) naturally induces a polynomial function on the space 
  $\mathsf{H}_{n+1}$ of Hermitian $(n+1)\times(n+1)$ matrices, and that $p$ (respectively $q$) is
  uniquely determined by the restriction of the respective polynomial function to an open subset of  $\mathsf{H}_{n+1}$.
  
  Now recall that by the Iwasawa decomposition, every $g \in \Gl(n+1)$ has a unique decomposition 
  of the form 
  \[
     g = T u \  ,
  \]    
  where $T$ is an upper triangular matrix with positive entries on the diagonal and $u \in \U(n+1)$.
  Because the map $g \mapsto (g^*)^{-1}$ is an automorphism of $\Gl(n+1)$ switching upper triangular and lower triangular   
  matrices and leaving the group $\U(n+1)$ invariant, any $g\in \Gl(n+1)$ can also be represented in the form
  \[
    g = Lu^\prime \  ,
  \] 
  where $L$ is a lower triangular matrix with positive diagonal entries and $u^\prime\in \U (n+1)$.
  These observations entail that the spaces $Y \U(n+1)$ and $Y^\prime \U(n+1)$ both contain the 
  group $\Gl (n+1)$ which is dense in $\Mat_{n+1}$. One concludes by  Equations 
  \eqref{eq:invariant-polynomial-representation-upper-triangular}  and 
  \eqref{eq:invariant-polynomial-representation-lower-triangular}
 that the equality
 \begin{equation}\label{eq:equality-polynomials}
    (\pi_{n+1}^{11} )^s \, p \big( (\pi_{n+1}^{ij})_{1\leq i,j\leq n+1} )\big)
    = (\pi_{n+1}^{(n+1)(n+1)} )^r   \, q \big( (\pi_{n+1}^{ij})_{1\leq i,j\leq n+1} )\big) 
 \end{equation}
 holds on $\Mat_{n+1}$. Since the image $\poshermitian_{n+1}$ of the map $\pi_{n+1} : \Mat_{n+1} \to \mathsf{H}_{n+1}$ contains the
 positive definite Hermitian matrices or in other words the invertible positive semi-definite Hermitian matrices, 
 the image $\poshermitian_{n+1}$ has open interior, and Equality \eqref{eq:equality-polynomials} holds on $\mathsf{H}_{n+1}$ that is
 \[
 \begin{split}
    x_{11}^s \, & p \big( (x_{ii})_{1\leq i \leq n+1} ,  (z_{ij})_{1 \leq i < j \leq n+1},  (\overline{z}_{ij})_{1\leq i < j \leq n+1} \big) 
    = \\ 
    & = x_{(n+1)(n+1)}^r \, 
    q \big( (x_{ii})_{1\leq i \leq n+1} ,  (z_{ij})_{1 \leq i < j \leq n+1},  (\overline{z_{ij}})_{1\leq i < j \leq n+1} \big)     \  ,
 \end{split}
 \]  
 where
 \[ 
   z_{ij} = x_{ij} + \imi y_{ij}\quad\text{and} \quad  
   \overline{z}_{ij} = x_{ij} - \imi y_{ij} \quad \text{for } 1\leq i < j \leq n+1 \ ,
 \]
  and where the variables
  \[ 
    x_{11}, \ldots , x_{(n+1)(n+1)}\ , \: (x_{ij})_{1\leq i<j\leq n+1} \text{ and }
   (y_{ij})_{1 \leq i < j \leq n+1}
  \]
  form the real coordinates of the space $\mathsf{H}_{n+1}$ of Hermitian matrices. Hence the polynomial $p$ is divisible 
  by $x_{(n+1)(n+1)}^r$ and the polynomial $q$ by  $x_{11}^s$. Let 
  \[
     P = \frac{p}{x_{(n+1)(n+1)}^r}  = \frac{q}{x_{11}^s}  
  \]   
  be the corresponding quotient polynomial. By construction,
  \[
    f = P \big( (\pi_{n+1}^{ij})_{1\leq i,j \leq n+1}\big) \ , 
  \] 
  and $\pi_{n+1} : \Mat_n \to \mathsf{H}_{n+1} \subset \Mat_{n+1}$ is a Hilbert map. Minimality is obvious because
  $\pi_{n+1}$ has $(n+1)^2$ real components and the real dimension of the open stratum of the orbit space 
  $\poshermitian_{n+1} \cong X_{n+1}$ is $(n+1)^2$ as well.  This proves the claim for $n+1$.
  \end{proof} 
 
 \begin{Remark}
 \begin{enumerate}[1.]
 \item   
  The idea to use a divisibility argument in the proof, more precisely the observation that the polynomial $p$ is 
  divisible by  $x_{(n+1)(n+1)}^r$ and the polynomial $q$ by  $x_{11}^s$, goes back to Weyl's book 
   \cite{Weyl.TheClassicalGroups} on classical groups, where an analogous divisibility argument was used in 
   Section II.2 to construct a Hilbert basis of polynomials  invariant under the left $\mathsf{O}(2)$-action on $\RR^2$.  
   Note that such an argument is non-trivial since here one has to work with real algebraic varieties and groups 
   where Hilbert's Nullstellensatz does not necessarily hold.
 \item   
   Let us also remark that to our knowledge the above construction of a Hilbert basis of invariant polynomials for the 
   right (or left) action of $\U(n)$ on $\Mat_n$ is new and not contained in the literature on invariant theory where usually 
   only complex algebraic group actions or the action of $\U(n)$ by conjugation are considered.  
  \end{enumerate}
 \end{Remark}

\begin{Corollary}
\label{Corollary1}
  The projection
  \[ \pi_n: \Mat_n \to \poshermitian_n , \quad a  \mapsto a a^* \]
   coincides with the orbit map of the right $\U(n)$-action on $\Mat_n$.  The stratifications of $\Mat_n$ and  $\poshermitian_n$ by orbit type  coincide with the respective 
   stratifications by rank and are both {\rm (}b{\rm )}-regular in the sense of Whitney.   Finally, $\pi_n$ becomes a proper surjective stratified submersion when $\Mat_n$ and 
   $\poshermitian_n$ are stratified by  rank. 
 \end{Corollary}
\bpf

One concludes by Theorem \ref{thm:HilbertMap} that the polynomials $\pi_n^{ij}$ form a Hilbert basis  of the invariant polynomials on $\Mat_n$. 
In particular this entails the already made observation that  the image $\poshermitian_n$ of $\pi_n$ coincides with the orbit space of the $\U(n)$-action. Moreover, the images $\poshermitian_n (k)$ 
of the strata $\Mat_n(k) = (\Mat_n)_{(U(n-k))} $ of $\Mat_n$ under the projection $\pi_n$ are manifolds, and the induced decomposition of the  orbit space
$\poshermitian_n$ is a {\rm (}b{\rm )}-regular stratification according to \cite[Thm.~4.4.6]{PflaumBook}. 
\epf

\begin{Corollary}
\label{Corollary2}
 The space of $n\times n$ density matrices $\densityops_n \subset \Mat_n$ is naturally stratified by 
 rank. This  stratification is {\rm (}b{\rm )}-regular.
\end{Corollary}

\bpf
 As observed before, $\densityops_n$ is the intersection of the semialgebraic   set $\poshermitian_n$ with the 
 linear subspace $V = \{ a \in \Mat_n \mid \tr (a) =1\} \subset \Mat_n $. The subspace $V$ has codimension $1$ and 
 intersects the  submanifold $\Mat_n(k)$ transversally, hence the stratification of  $\poshermitian_n$  by rank 
 induces a stratification of $\densityops_n$. It fulfills Whitney's condition (b) by construction and since the stratification 
 by rank on $\poshermitian_n$ is {\rm (}b{\rm )}-regular. 
 \epf

%


\section{Uhlmann's connection}
\label{UhlmannConn-finD}


Here, we recall Uhlmann's construction of a natural connection on the space of density matrices in the
finite-dimensional case. First, we define Uhlmann's connection on the principal stratum. 
For that purpose, denote the Riemannian metric on the group $\Gl(n, \C)$ induced from the natural scalar product 
\eqref{NatScPr-Mn} on $\Mat_n$ by $h$, that is, let 
\begin{equation}
\label{RiemMetr-GlnC}
h (X,Y) = \Real \tr(X^\ast Y) \quad \text{for all $X, Y \in \tg_g \Gl(n, \C) \cong \Mat_n$ with $g\in \Gl(n, \C)$} \ .
\end{equation}
Uhlmann's equation
\begin{equation}
\tag{\ref{UhlmannEqn}}
a^\ast \d a -(\d a^\ast) a = a^\ast a {\bf A} + {\bf A} a^\ast a 
\end{equation}
then implicitely defines a connection $1$-form ${\bf A}$ on the generic stratum, that is, on the principal bundle
\eqref{Fibr-M-D-princStr}. In the following we make this more precise and explain how Uhlmann's equation needs to be
interpreted.

\begin{Theorem}
\label{S-UhlmConn-PrincStr}
On the principal bundle $\pi_{n,n}: \Gl(n, \C) \to  \poshermitian_n(n) $ there exists a unique connection
$1$-form
\[ {\bf A} :  \Gl(n, \C)\to\tg^*\Gl(n, \C)\otimes\mathfrak{u} (n)\cong\Gl(n, \C) \times \big(\Mat_n^* \otimes \imi \hermitian_n\big)\]
which for every  $g\in \Gl(n, \C)$ and $X \in \tg_g\Gl(n, \C)$ solves the equation
\begin{equation}
\label{UhlmannEqn-princStr}
  g^\ast X - X^\ast g  =
  g^\ast g \big(X \, \lrcorner \, {\bf A} \big) + \big( X \, \lrcorner \, {\bf A} \big) g^\ast g \ .
\end{equation}
Moreover, the following properties are equivalent:
\begin{enumerate}[{\rm (1)}]
\item \label{ConnItem1}
  The tangent vector $X $ is horizontal with respect to the connection $1$-form $\bf A$, that is
  \[ {\bf A}( X) = 0 \ . \]
\item \label{ConnItem2}
 The tangent vector $X$ fulfills 
\beq
\label{UhlmConn-HorC}
g^\ast X - X^\ast g = 0 \, .
\eeq
\item \label{ConnItem3}
  With respect to the Riemannian metric $h$, the tangent vector $X$ is orthogonal to the vertical subspace
  \[ {\mr V}_g := \ker {\mr T}_g \pi_{n,n} \subset \tg_g \Gl(n, \C)\ .\]
\end{enumerate}
Finally, the distribution $ {\mr H} = \bigsqcup_{g \in \Gl(n,\C)}  {\mr H}_g$ with  
\begin{equation}
\label{HorDistr}
  {\mr H}_g = \left\{ X \in \tg_g \Gl(n,\C)\mid g^\ast X - X^\ast g = 0  \right\} \quad \text{for all } g\in \Gl(n,\C) 
\end{equation}
is a connection on the bundle $\pi_{n,n}$, that is, $ {\mr H}$ is a $\U(n)$-right invariant subbundle of
$\tg\Gl(n,\C)$ and the  equality
\begin{equation}
\label{Splitting}  
   \tg \Gl(n,\C) = {\mr V} \oplus {\mr H} 
\end{equation}
holds true, where $\mr V = \bigsqcup_{g\in \Gl(n,\C)} \mr V_g$ is the vertical subbundle of $\tg\Gl(n,\C)$.
The connection $1$-form corresponding to $\mr H$ coincides with the $1$-form $\bf A$. 
\end{Theorem}
    
\begin{Remark}
  Equation \eqref{UhlmannEqn-princStr} in the theorem indicates how Uhlmann's equation \eqref{UhlmannEqn} needs to be
  read. If one interprets the symbol $a$ in there as standing for the identity map on $\Gl (n,\C)$, then inserting a 
  tangent vector $X \in \tg_g\Gl(n,\C)$ in both sides of \eqref{UhlmannEqn}  results exactly in Equation
  \eqref{UhlmannEqn-princStr}. 
\end{Remark}

\begin{proof}
  Let $A$ be an element of the Lie algebra $\mathfrak{u} (n)\cong \imi \hermitian_n  $ of anti-Hermitian
  $n\times n$ matrices. The fundamental vector field $A_{\Gl (n,\C)}$ then maps every invertible $g$ to
  the vector $gA$. This entails that the vertical space $\mr V_g$ is given by
  \begin{equation}
  \label{VertCond}
     \mr V_g = \{ Y \in \mr T\Gl(n,\C)\mid g Y^\ast + Y g^\ast = 0 \}   \ .
  \end{equation}
  Thus, a tangent vector $X\in \Mat_n$ at $g$ is orthogonal to ${\mr V}_g$ if and only if
  \[
    \tr\big(A (X^\ast g - g^\ast X)\big)  = 2 \Real \tr(X^\ast gA) = 2 h (X, gA) = 0
    \quad \text{for all } A \in \imi \hermitian_n \ .
  \]
  Since $X^\ast g - g^\ast X$ is anti-Hermitian, $X$ is orthogonal to  ${\mr V}_g$ if and only if
  \[ X^\ast g - g^\ast X = 0  \ . \]
  This proves the equivalence of \eqref{ConnItem2} and \eqref{ConnItem3}. Moreover, it shows that the distribution
  $\mr H$ defined by Equation \eqref{HorDistr} is indeed the horizontal bundle of the principal connection
  defined by the $\U(n)$-invariant Riemannian metric $h$. In particular, $\mr H$ is complementary to the canonical
  vertical distribution $\mr V$ which means that  Equation \eqref{Splitting}  holds true.

  The connection $1$-form associated to the horizontal distribution $\mr H$ is given over the footpoint
  $g\in \Gl (n,\C)$ as the  orthogonal projection onto the horizontal space followed by the natural isomorphism
  between ${\mr H}_g$ and the Lie algebra $\mathfrak{u} (n)$: 
  \[
    {\bf A}_g : \tg_g\Gl (n,\C) = {\mr V}_g \oplus {\mr H}_g \to \mathfrak{u} (n) , \quad (gA, X) \mapsto A \ .
  \]
  By construction, the equivalence of \eqref{ConnItem1} and \eqref{ConnItem3} follows. Let us show that
  the thus defined $\bf A$ satisfies Equation \eqref{UhlmannEqn-princStr}.
  If $Y = gA$ with $A \in \mathfrak{u}(n)$ is a vertical vector over the footpoint $g$, then
  $Y \lrcorner {\bf A} = A$ and \eqref{UhlmannEqn-princStr} holds true since $A$ is anti-Hermitian.
  In case $X \in {\mr H}_g$, Equation \eqref{UhlmannEqn-princStr} holds   by equivalence of
  \eqref{ConnItem1} and \eqref{ConnItem2}. It remains to show that  Equation \eqref{UhlmannEqn-princStr}
  uniquely determines the $\mathfrak{u}$-valued form $\bf A$. To this end note that one can equivalently write
  this equation in the form
  \[
    (g^* g)^{-1} \big( g^*X - X^* g \big) = \big( I + \operatorname{Ad} (g^*g) \big) (X \lrcorner {\bf A}) \ , 
  \] 
  where $I$ is the identity. In case  $I + \operatorname{Ad} (g^*g) : \mathfrak{u} (n) \to \mathfrak{u} (n)$
  is injective, the claim follows. But that is the case, indeed, because if
  $A \in \ker \big( I + \operatorname{Ad} (g^*g)\big) $, then the vertical vector  $Y = gA$ 
  fulfills $g^*Y - Y^*g =0$, hence is also horizontal, and therefore $A$ has to vanish.     
\end{proof}

Uhlmann's equation \eqref{UhlmannEqn} immediately implies the transformation law 
\[
   R_u^* {\bf A} = u^\ast {\bf A} u = \Ad (u^\ast) {\bf A}
\]
under right translations $R_u : \Gl (n,\C) \to \Gl(n,\C)$, $g \mapsto gu$ for any $u \in \U(n)$.
Finally, the reader also easily checks that one has the standard gauge transformation law
\[
  \alpha^* {\bf A} = \Ad (u) {\bf A} + u^{-1} \d u 
\]
for any vertical automorphism $\alpha :  \Gl (n,\C) \to \Gl(n,\C)$, $g \mapsto g \, u (g)$
where now $u :  \Gl (n,\C) \to \U (n)$ is a smooth map.

Next, let us analyze the defining  equation \eqref{UhlmannEqn} for the singular strata. In  analogy to Theorem
\ref{S-UhlmConn-PrincStr}, we obtain the following.
\bsz
\label{S-UhlmConn-singStr}
For every natural rank $k \leq n$, the distribution  $ {\mr H}_{ \Mat_n(k)} = \bigsqcup_{a \in  \Mat_n(k)}  {\mr H}_a$ with
\beq
\label{HorDistr-sing}
{\mr H}_a = \left\{ X \in \tg_a \Mat_n(k) \mid a^\ast X - X^\ast a = 0  \right\}
\eeq
defines a connection on the bundle $\pi_{n,k}$. 
\esz

\bpf
Consider the restriction of the defining condition in \eqref{HorDistr-sing} to $\Mat_{\U(n-k)} \subset \Mat_n(k)$. Since 
${\rho}_{n,k} : \Mat_{\U(n-k)} \to \poshermitian_n(k)$, $a \mapsto aa^\ast$ is a $\U(k)$-principal bundle, by the same arguments as in  the proof of Theorem \ref{S-UhlmConn-PrincStr}, the restriction of $ {\mr H}_{ \Mat_n(k)}$ to
$\Mat_{\U(n-k) }$ defines a connection ${\mr H}_{\U(n-k)}$ on  ${\rho}_{n,k}$.  Now, by standard arguments, this connection defines a connection on the associated bundle $E_n(k)$ and the latter induces via $\Phi$ given by \eqref{Iso-Phi} a connection on $\pi_{n,k}$. 
To determine this connection, define for $u \in \U (n)$ the mapping
$$
\iota_{[u]}: \Mat_{\U(n-k)} \to E_n (k) \, , \quad \iota_{[u]} (a) := [(a,  [u])]  
$$
and consider the composition 
\beq
\label{Map-MUn-k-Mnk}
\Phi^\prime_{[(a,  [u])]} \circ (\iota_{[u]})^\prime_{a} : \tg_a \Mat_{\U(n-k)} \to \tg_{ a u} \Mat_n(k) \, .
\eeq
Given $a \in \Mat_{\U(n-k)}$ and $X \in \tg_a \Mat_{\U(n-k)}$, the image of $X$ at 
$\Phi([(a, [u])]) = au$ 
under the mapping \eqref{Map-MUn-k-Mnk} is $X u \subset \tg_{au}  \Mat_n(k)$. Next, every vertical vector at
$au$ is of the form
\beq
\label{vertV-Mn-k}
  Y = a A u  \quad\text{with }  A \in {\mf u}(k) \, .
\eeq
Since $\U(k)$ and $\U(n-k)$ commute, we have for every $v \in \U(n-k)$
\[ a A vu  =  a v A u = a A u  \, \]
showing that $ a A u  $ is unambiguously defined.
Thus, $X  \in \tg_{a u}  \Mat_n(k)$ is orthogonal to the vertical subspace 
${\mr V}_{au}$ if and only if
\[
  \tr(X^\ast a A u  ) = 0 \, .
\]
Using  $ aA u  =  a u u^\ast A u $ and $ u^\ast A u = \Ad(u^{-1}) (A) \in {\mf u}(n)$, we can rewrite this condition 
as follows:
\beq
\label{OrthC-Mn-k}
\tr \left( \Ad(u^{-1}) (A) (X^\ast a Au -  (aA u )^\ast  X ) \right) = 0 \, .
\eeq
This implies that $a A u $ is horizontal if and only if
\[
 X^\ast a A u  -  (a A u )^\ast  X  = 0 \,.
\]
Thus, the assertion follows.  
\epf

\begin{Proposition}\label{prop:uhlmann-connection-stratified-connection}
  The union ${\mr H} = \bigsqcup_{0 \leq k \leq n} {\mr H}_{ \Mat_n(k)} $ of the Uhlmann connections over the rank $k$ strata
  is a stratified vector subbundle of $T\Mat_n$ and constitutes a stratified connection which we also call \emph{Uhlmann connection}. 
\end{Proposition}
\begin{proof}
Observe that the condition $a^*X-X^*a$ is continuous in $a$ and $X$.
That is, given a sequence $(X_k)_{k\in\N}$ with respective basepoints $(a_k)_{k\in\N}$, if
$a_k^* X_k - X^*_k a_k = 0$ and if $X_k$ and $a_k$ have limits $X$ and $a$ in $T \Mat_n$ and in $\Mat_n$, respectively,
then $a^*X-X^*a=0$ and so the limit $X$ is again in the kernel of the corresponding connection $1$-form.
\end{proof}

Next, let us take up the discussion in Remark \ref{B-GeomStrata}.
\bbm
\label{GeomConn-sing}
On the singular strata, Uhlmann's connection has the following geometric interpretation. 
\ben
\item\label{Ite1-GeomConn-sing}
  By \eqref{Emb-Pn-Mn}, $P_n(k)$ may be regarded as a subbundle of $\pi_{n,k}$. Thus, using $p p^\ast = \id_{\C^k}$,
  we obtain the following restriction of \eqref{Emb-Pn-Mn} to $P_n(k)$:
\begin{equation}
\label{UhlmannConn-sing}
  f^\ast \d f - (\d f^\ast) f = f^\ast f {\bf A} + {\bf A} f^\ast f \, .
\end{equation}
Now, by the same arguments as in the principal stratum case, one shows that $\bf A$ is a $ {\mf u}(k)$-valued connection form 
on the principal bundle $\lambda_{n,k}$. 
\item\label{Ite2-GeomConn-sing}
  By Point \ref{Ite2-B-GeomStrata}. of Remark \ref{B-GeomStrata}, the Stiefel bundle $\overline{\lambda}_{n,k}$ is a
  subbundle of $\pi_{n,k}$. By \eqref{Emb-Sn-Pn}  and \eqref{Emb-Pn-Mn}, every tangent vector
  $X \in  \tg_a i_{n,k} (S_n(k))$  fulfills
  \[
    X^\ast a + a X^\ast = 0\quad \text{for all } a \in  i_{n,k} (S_n(k)) \ .
  \] 
  Using this, the horizontality condition in \eqref{HorDistr-sing} implies that $X$ is horizontal if and only if
  $a^\ast X = 0$. From this condition, we read off the connection form
\beq
\label{UhlmConn-StiefelB}
   {\bf A} =  a^\ast \d a \, .
\eeq
Its pullback under $i_{n,k}$ yields the canonical connection of the Stiefel bundle.
\een
\ebm

It remains to implement the normalization condition 
\beq
\label{NormCond}
\tr(a a^\ast) = 1\, .
\eeq 
If we impose it, then $\pi_n$ is given by 
\eqref{Fibr-W-Omega} and we have, additionally to the action of $\U(n)$, the action
\[
  \R_+ \times \Mat_n \to \Mat_n\, , \quad (\lambda, w) \mapsto \lambda \cdot w \, ,
\] 
of the multiplicative Lie group $\R_+$ on the fibers of $\pi_n$. Let us denote the space of elements
$a \in \Mat_n$ fulfilling \eqref{NormCond} by  $S \Mat_n$ or $\hat\Mat_n$ and, as before, the space of (normalized) density matrices
by $\densityops_n$. We clearly have
\[
  S\Mat_n = S^{2 n^2 -1} 
\]  
and, thus, by implementing the normalization condition, we obtain the projection 
\[
  \hat \pi_n : S^{2 n^2 -1} \to \densityops_n \ .
\]
Restricting $\hat\pi_n$ to $ \hat\Mat_n (n) := S\Mat_n \cap \Gl (n,\C)$, one obtains the $\U(n)$-principal bundle
\[
  \hat \pi_n (n) : \hat\Mat_n (n) \to \densityops_n (n) \ ,
\]
where $\densityops_n (k)$ denotes the space of rank $k$ density matrices. 
The principal bundle $\hat \pi_n (n)$ is endowed with the ${\mf u}(n)$-valued connection form $\bf A$
defined by \eqref{UhlmannEqn}, with $a $ fulfilling \eqref{NormCond}. 

Next, let us discuss the singular strata. Here, inserting the right hand side of
\eqref{Struct-M-n-k} into \eqref{NormCond}, one obtains
\beq
\label{NormCond-sing}
  \tr(f f^\ast) = 1 \, .
\eeq
Thus, we obtain the bundles $\hat \pi_n(k): \hat \Mat_n(k) \to \densityops_n (k)$, associated with the principal bundles 
\[
  \hat \lambda_{n,k}: \hat P_n(k) \to \densityops_n(k) \, ,
\] 
where the embedding  $\hat  P_n(k) \subset  P_n(k) $ is defined by \eqref{NormCond-sing}. Correspondingly,    $\hat  P_n(k)$ 
is endowed with the ${\mf u}(k)$-valued connection form $\bf A$ defined by \eqref{UhlmannConn-sing}, where $f$ fulfills the normalization condition \eqref{NormCond-sing}. Note that for the special case $k = 1$ the normalization condition 
\eqref{NormCond-sing} coincides with the embedding condition \eqref{Emb-Sn-Pn} and, thus, the bundle
$\hat \pi_{n,1}$ is associated with the Stiefel bundle $\overline{\lambda}_{n,1}$. Consequently, on that stratum the
${\mf u}(1)$-valued connection form $\bf A$ coincides with the canonical connection form of the Stiefel bundle
as given by \eqref{UhlmConn-StiefelB}, cf.\ Point \ref{Ite2-GeomConn-sing} of Remark \ref{GeomConn-sing}. Finally, note that by implementing the normalization condition we arrive at topologically non-trivial bundle 
structures. 

\bbs
\label{B-SpecialCase}
Consider the special case $n = 2$. Here, $\hat \Mat_2(2) = S^7$ and $\densityops_2(2)$ coincides with the interior of the  $3$-dimensional disc ${\mr D}^3$. There is one singular stratum
$\hat \pi_{2,1}: \hat \Mat_2 (1) \to \densityops_2(1) $, where
\[ 
  \hat \Mat_2 (1) = S_2(1) \times_{\U(1)} \U(1) \backslash \U(2) \cong S^3 \times_{\U(1)} S^3
\]
and $\densityops_2(1) $ coincides with the boundary $S^2$ of ${\mr D}^3$. By the above discussion, the
connection $\bf A$ on the singular stratum is given by the canonical connection of the Stiefel bundle.
For $n= 2$, the latter coincides with the canonical Hopf bundle connection (gauge potential of a magnetic monopole). Thus, on the singular stratum we know $\bf A$ explicitly. 
For $n = 2$, the connection $\bf A$ on the generic stratum may be explicitly determined, too. This was accomplished in \cite{DR-2}. Moreover, if one endows the space of density matrices with 
the Riemannian structure induced by the Bures metric, then $\bf A$ can be shown to fulfil the source-free Yang-Mills equation for the generic stratum and 
the source free Maxwell equation for the singular stratum. Finally, we note that the connection on the generic stratum turns out to be closely related to the canonical connection on the quaternionic Hopf bundle (instanton bundle), see \cite{DR-1} for the details. 
\ebs

To close this section, let us mention that there exists a natural class of connections on the above described bundles containing Uhlmann's connection as a distinguished special case. This was shown in \cite{DR-2}.

%

\section {Comments and Outlook}
\label{Outlook}


Recall the evolution equation \eqref{EvEqu-DensM}. Let us assume that the quantum Hamiltonian $\hat H$ does not depend on $\rho$. In that case, this equation may be studied as follows: we endow the space of density matrices with its classical Lie-Poisson bracket,\footnote{Behind, there is the dual space ${\mr u}^\ast (n)$ of the unitary Lie algebra, endowed with the minus-bracket.}
$$
\{ f,g \} (\rho) = - i \tr \left( \rho \cdot [ \nabla f (\rho), \nabla g(\rho)]  \right) \, ,
$$
and define the classical Hamiltonian function $H(\rho) := \frac{1}{\hbar} \tr( \rho \hat H )$. Then,  $\nabla H = \frac{1}{\hbar} \hat H$ and thus the evolution equation \eqref{EvEqu-DensM} matches the classical Hamiltonian 
equation $\dot{\rho} = \{\rho, H\} $ and we obtain the following Hamiltonian vector field on the stratified space:
$$ 
X_H (\rho) = \frac{1}{i \hbar} [{\hat H}, \rho ] \, .
$$  
This shows that, for $\hat H$ not depending on $\rho$,  the quantum dynamics governed by \eqref{EvEqu-DensM} may be studied in 
terms of the Hamiltonian vector field $ X_H $ on the stratified space of density matrices.\footnote{There is a huge literature on the 
subject, see e.g. \cite{Kibble}, \cite{Brody}, \cite{Grabowski}, \cite{Sinha} and references therein.} For that purpose,  the 
Whitney {\rm b}-regularity of our theory proved in Section \ref{Space-densMatr} turns out to be significant; it guarantees 
a nice dynamical  behaviour of $  X_H $ and of its flow.  

To illustrate this, let us consider a differentiable stratified space $(X,\calC^\infty(X))$ and let us discuss the behaviour of a vector field and of its integral curves in the vicinity of a singularity (a lower stratum). We use the notation introduced in  Definition 
\ref{Def-WhitneyReg}, where Whitney (\rm b)-regularity is defined. As the stratified space may be regarded as embedded in $\R^n$, we may freely use the terminology of Euclidean geometry. Let $R$ and $S$ be two strata belonging to $(X,\calC^\infty(X))$ such that 
$R \subset \overline{S}$.  To start with, we construct a tubular neighbourhood $U \subset S$ around $R$ by introducing the projection map 
$\pi: U \to R$ given by assigning to any $y \in U$ the unique closest point $x \in R$, that is $\pi(y) = x$. Then, the tubular neighbourhood diffeomorphism is given by
$$
\chi: U \to {\mr N} R \, , \quad \chi(y) = (\pi(y), y - \pi(y) ) \, .
$$
The second component is the secant vector living in the normal fibre ${\mr N}_x R$, that is, it is orthogonal to $\tg_x R$. Moreover, we define the distance function 
\beq
\label{DistFunction}
\delta (y) := \parallel y - \pi(y) \parallel^2 \, .
\eeq
Now,  Whitney (\rm b)-regularity implies the following. 
\ben
\item
Let $\xi$ be vector field  on $(X,\calC^\infty(X))$.
If an integral curve $t \mapsto \gamma(t)$ of $\xi$  in the stratum $S$  approaches 
a boundary point $x_0$ of the lower-dimensional stratum $R$, then this must happen along a well-defined asymptotic direction in the projective space. Indeed, by definition, the tangent space $\R  \cdot \dot{\gamma}(t)$ is a subspace of the moving tangent space 
$\tg_{\gamma(t)} S$. For a sequence $t_k \to \infty$, the sequence of $ y_k = \gamma (t_k) $ approaches  $x_0$. 
Let the sequence of secants $\ell_k = (y_k, x_0) $ converge to the line $\ell$ and the tangent spaces $\tg_{y_k} S$ converge to the 
subspace $T \subset \R^n$. Then, by (\rm b)-regularity, we have $\ell \subset T$, that is, the directions of approach must stay contained in the limiting tangent spaces. As $\xi$ is smooth in the ambient space, the velocity vector $\dot{\gamma}$ aligns asymptotically with the secant line $\ell$. 

In order to show that the  vector field $\xi$ undergoes a continuous transition, when moving from the upper stratum $S$ toward the lower stratum $R$,   Whitney (\rm b)-regularity is necessary, because it guarantees that the tangent spaces  $\tg_{y_i} S$ approach $T$ (containing $\tg_{x_0} R$) smoothly. However, this may not be sufficient. We need some further assumption, e.g. we can assume that our space is a stratified Poisson manifold. We are not going into further details, here.

\item
An integral curve $t \mapsto \gamma(t)$ of $\xi$ starting in the interior of the upper stratum $S$ cannot reach the singular stratum $R$ in a finite amount of time. Indeed, by \eqref{DistFunction}, we have
$$
\dot{\delta}(t) = \langle \nabla \delta (\gamma(t)), \dot{ \gamma}(t) \rangle 
= \langle \nabla \delta (\gamma(t)), \xi (\gamma(t)) \rangle \, .
$$
As the gradient $\nabla \delta (\gamma(t))$ is a vector pointing along the secant line connecting $y$ with its closest point $\pi(y) \in R$, its norm scales proportionally to the distance, that is, $\parallel \nabla \delta (y) \parallel \leq c_1 \sqrt{\delta (y)}$. Now, by Whitney 
(\rm b)-regularity, the secant line $\nabla \delta (y)$ cannot become orthogonal to $\tg_y S$ and, thus, the projection of $\xi$ onto this secant line must be bounded by the distance function itself, that is, $|\langle \nabla \delta (y), \xi(y) \rangle | \leq c_2 \delta (y)$. This yields 
\beq
\label{GroenwallInequality}
| \dot{\delta}(t) | \leq c_2 \delta(t) \, .
\eeq
This implies an infinite arrival time, because integrating the resulting lower bound following from \eqref{GroenwallInequality} yields 
$\delta(t) \geq \delta(0) e^{-c_2 t} $, which is an example of Gr\"onwall's Inequality \cite{Groenwall}. Note that integrating the 
resulting upper bound yields bounds on the speed at which $\delta $ can grow or decay. 
\item
As a consequence of Definition \ref{Def-WhitneyReg}, one has the controlled lift of vector fields from $R$ to $S$ of Mather.\footnote{This is the basic input for proving the famous Thom-Mather Isotopy Theorem.} 
Let $\zeta$ be a smooth vector field on $R$. Then, one defines its lift $\hat \zeta$ of 
$\zeta$ to $S$  by  pointwise projecting $\zeta ( \pi (y ))$ onto $\tg_y S$ in such a way that 
$\delta^\prime \big({\hat \zeta}(y) \big) = 0$, 
that is, in such a way that ${\hat \zeta}(y) $ is orthogonal to the secant line.\footnote{In order to fulfil this condition, one has to add to the standard projection a certain correction term, which we are not going to spell out here.} Thus, by construction, the lift so defined fulfils the following control  conditions: 
\beq
\label{ControlCond}
\pi^\prime \big( {\hat \zeta}(y) \big) = \zeta (\pi(y) ) \, , \quad \delta^\prime \big({\hat \zeta}(y) \big) = 0 \, ,
\eeq
for any $y \in U \subset S$. To summarize, Whitney  (\rm b)-regularity allows one to extend vector fields on $R$ into controlled tubular neighbourhoods. The second condition in \eqref{ControlCond} is to be regarded as a distance preservation condition. It guarantees that 
the flow of the lifted vector field slides parallel to the lower stratum, thus, preserving the transversal slices as one moves along $R$. 
The controlled lift existence implies a number of differential-topological consequences, which we cannot discuss here. 
\een

Finally, let us recall the following remark made by Uhlmann in   \cite[Sec.4]{AU-1}:
\begin{quote}
\textit{In order to bypass complicated technicalities, I restrict myself to
finite-dimensional Hilbert spaces although this assumption does not show up
explicitly and is probably inessential.}
\end{quote}
Up to our knowledge, the extension of the above picture to infinite-dimensional Hilbert spaces was never accomplished. 
On the other hand, some modern developments -- in particular, the study of infinite-dimensional lattice systems -- suggest that it is worthwile to study the infinite-dimensional case in detail. This brings us to the following program for future work:
\begin{enumerate}
\item
Study the geometry of the stratification defined by $\pi$ in the infinite-dimensional case.
\item
Investigate the infinite-dimensional counterpart of Uhlmann's connection $\bf A$ including its geometric meaning on the 
non-generic strata of $\pi$.
\item
Study the evolution equation \eqref{EvEqu-DensM} and, in particular, parallel lifts of solutions of this equation.
\end{enumerate}
Our first goal will be to show that the space of density matrices $\calD$ on a (separable)
Hilbert space $\calH$ possesses a natural stratification or in other words decomposition
into smooth, possibly infinite dimensional Banach manifolds. It turns out the novel stratification theory for spaces of density operators  
described in the appendix may be extented to the  infinite-dimensional  context. This will be the first corner stone for accomplishing this 
program.

\appendix


\section{Whitney stratifications}
\label{appendix}


For the convenience of the reader, we outline here the fundamentals of stratification theory, focusing specifically on the Whitney conditions.
We restrict our attention to (a) and (b)-regularity of orbit spaces of compact group actions on manifolds.
For a comprehensive treatment and detailed proofs of the Thom-Mather 
stratification theory we refer the reader to {\sc Mather} \cite{MatNTS,MatSM} and to {\sc Pflaum} \cite{PflaumBook}.
For the stratification theory of orbit spaces we also refer to {\sc Bierstone} \cite{BierstoneOrbitspaces}
and {\sc Duistermaat--Kolk} \cite{DuistermaatKolk}.

\begin{Definition}
  Let $X$ be a second-countable paracompact Hausdorff space. By a \emph{manifold decomposition} or a
  \emph{prestratification} of $X$ 
  one understands a collection $\calZ$ of locally closed non-empty subspaces $S\subset X$ called \emph{pieces} such that
  the following conditions hold true:
  \begin{enumerate}[{\rm (DS1)}]
  \item\label{ite1:decomposition} Each piece $S$ is a topological manifold of some finite dimension $\dim (S)$.
  \item The collection  $\calZ$ is a partition of $X$ that means $X$ coincides with the union  $\bigcup_{S\in \calZ} S$
    and for all $S_1,S_2\in \calZ$ one either has $S_1 =S_2$ or $S_1\cap S_2 =\emptyset$.
  \item\label{ite:local-finiteness}
    The partition $\calZ$ is locally finite that means for every point $x \in X$ there exists a neighborhood
    $U\subset X$ which meets only finitely many elements of  $\calZ$.
 \item\label{ite:condition-frontier}
   The \emph{condition of frontier} is fulfilled which means that for all $S,R\in \calZ$
   the relation $R\cap \overline{S}\neq \emptyset$ implies $R \subset \overline{S}$. One says in this situation that
   $R$ is \emph{incident} to $S$ and writes $R \leq S$. 
 \end{enumerate}
 A topological space $X$ endowed with a manifold decomposition will be termed a \emph{decomposed space}.
 We silently assume in this situation that $X$ is second-countable, paracompact and Hausdorff. 
\end{Definition}

It is a straightforward exercise to verify that the incidence relation on a manifold decomposition $\calZ$
of a topological space $X$ is an order relation. In other words,  $\calZ$ carries naturally the structure of a poset.
We therefore denote the incidence relation on  $\calZ$ by $\leq$. 
\begin{Remark}
  Condition (DS\ref{ite1:decomposition}) can be relaxed to allow for infinite dimensional manifolds like Banach, Hilbert or
  Fr\'echet manifolds.  In this paper we do not need this broader generality, but it will become necessary in future work where we
  plan to consider the stratification theory of state spaces of infinite dimensional Hilbert spaces. 
\end{Remark} 

\begin{Definition}
  Given decomposed spaces $(X,\calP)$ and $(Y,\calQ)$ one calls a continuous map $f:X\to Y$ a
  \emph{morphism of  decomposed spaces} if for all pieces $P\in\calP$ the image $f(P)$ is contained in a piece of
  the decomposition $\calQ$. In case $X=Y$ and the identity map on $X$ is a morphism of decomposed spaces from
  $(X,\calP)$ to $(X,\calQ)$ one calls $\calP$ \emph{coarser} than $\calQ$. 
  The coarser relation then is an order relation on the set of manifold decompositions of $X$.
\end{Definition}

Mather's approach to stratifications in \cite{MatSM} essentially consists in localizing decompositions.
Before we write down the precise definition let us  briefly recall that two subsets $A,B \subset X$ of a topological space $X$ are said to be \emph{equivalent at} $x\in X$
if there exists an open neighborhood $U$ of $x$ such that $A\cap U = B\cap U$. Obviously, this defines an equivalence
relation on the power set of $X$. We denote the equivalence class of $A$ with respect to this equivalence relation
by $[A]_x$ and call it the \emph{set germ} or the \emph{germ} of $A$ at $x$.

\begin{Definition}[{cf.\ \cite[I.2.]{MatSM} \& \cite[Def.~1.2.2]{PflaumBook}}]
  Let $X$ be a second-countable paracompact Hausdorff space. A \emph{stratification} of $X$ then is a map $\calS$ which assigns to each
  point $x\in X$ the germ $\calS_x$ at $x$ of a closed subset of $X$ such that the following holds:
  \begin{itemize}
  \item[{\rm (ST)}]
    For each $x\in X$ there is a an open neighborhood $U$ of $x$ and a manifold decomposition $\calP_U$ of $U$
    such that for each $y\in U$ the set germ $\calS_x$ coincides with the germ $[P(x)]_x$, where $P(x) \in \calP_U$ is the piece
    of the decomposition in which $x$ lies.     
  \end{itemize}
  A topological space endowed with a stratification will be termed a \emph{stratified space}. 
\end{Definition}

\begin{Example}
  A manifold decomposition $\calZ$ of $X$ defines a stratification by putting $\calS_x = [P(x)]_x$, where $P(x)\in \calZ$
  is the piece containing $x$. We call the stratification $\calS$ \emph{induced} or \emph{defined} by $\calZ$. 
\end{Example}

\begin{Propdef}[{cf.\ \cite[Prop.\ 1.2.7]{PflaumBook}}]
  Let $(X,\calS)$ be a stratified space. Then there exists a coarsest decomposition defing the stratification.
  Often this decomposition is also denoted by the same symbol $\calS$ as for the underlying stratification. The pieces $S$
  of this  decomposition are called the \emph{strata} of the stratified space $X$.
\end{Propdef}

\begin{Example}[{cf.\ \cite[Sec.\ 4.3]{PflaumBook}}]
  Let $M$ be a finite dimensional manifold endowed with a smooth right $G$-action by a compact Lie group $G$.
  For every point $p\in M$ denote by $G_p$ the isotropy group of $p$ that means the set of all $g \in G$ such that
  $pg =p$.  Denote by $(G_p)$ the conjugacy class of the subgroup $G_p$ and by $M_{(G_p)}$ the set of all points $q$ in $M$   
  whose isotropy group $G_q$ is conjugate to $G_p$. Assigning to each point $p\in M$ the set germ $\big[M_{(G_p)}\big]_p$ then 
  is a stratification of $M$ called the \emph{stratification by orbit type}. Moreover, if $X =M/G$ denotes the orbit space
  endowed with the quotient topology and $\pi : M \to X$ the canonical projection, then the map which assigns to each orbit $[p] = pG \in  X$ the
  set-germ $\big[ \pi( M_{(G_p)}) \big]_{[p]}$ is a stratification of $X$ which one also calls \emph{stratification by orbit type}.
  These stratification results for $M$ and the orbit space $X$ are both consequences of the classical slice theorem of Palais. 
  Note that in infinite dimensions such a slice theorem holds only in particular cases and that therefore stratification by orbit type in infinite dimensions is more subtle; see \cite{DiezRudolph} for details.  
\end{Example}

The last example is crucial for this paper and we will further elaborate it. Before we come to this let us briefly examine stratified
spaces (locally) embedded in euclidean space.

\begin{Definition}[{cf.\ \cite[Sec.\ 1.3]{PflaumBook}}]
  By a \emph{smooth structure} on a stratified space $(X\calS)$ we understand an algebra of continuous functions on $X$, usually denoted by $\calC^\infty (X)$,
  for which there exists a homeomorphism $\iota: X \to A $ to a closed subset $A \subset \R^n$
  such that the pull-back
  \[ \iota^*: \calC^\infty (\R^n) \to \calC^\infty (X) , \quad f \mapsto f \circ \iota \]
  defines an isomorphism $\calC^\infty (\R^n)/ \calJ_X \to \calC^\infty (X)$, where
  $\calJ_X = \{ f \in \calC^\infty (\R^n) | \, f|_X = 0\}$
  is the vanishing ideal on $X$, and such that for each stratum $S\subset X$ the restriction
  $\iota|_S : S \to A$ defines a homeomorphism onto a smooth submanifold $\iota(S) \subset \R^n$. 
  One then calls the pair $(X,\calC^\infty(X))$ a \emph{differentiable stratified space} and $\iota$ a \emph{chart} for
  it.  
\end{Definition}

\begin{Example}
  Let $X \subset \R^n$ be a real algebraic or a real analytic variety. Then $X$ possess a canonical
  decomposition into smooth algebraic respectively  real analytic submanifolds of the ambient $\R^n$ which endows
  $X$ with the canonical structure of  a differentiable stratified space. 
\end{Example}

\begin{Example}[{cf.\ \cite[Sec.\ 4.4]{PflaumBook}}]\label{ex:differentiable-stratified-space-structure-orbit-space}
  Again let $M$ be  a smooth manifold endowed with a right action by a compact Lie group $G$. 
  Denote by $X$ the orbit space of this right action and by $\pi : M \to X$ the canoncial projection or in other words
  the orbit map. Then the algebra $\big(\calC^\infty(M)\big)^G$ of $G$-invariant smooth functions on $M$ can be identified
  canonically with a  subalgebra $\calC^\infty (X) \subset \calC (X)$ of the algebra of continuous  functions on $X$
  via the push-forward
  \[ \pi_* :   \big(\calC^\infty(M)\big)^G \to \calC (X), \quad f \mapsto \pi_*f := \big( X \ni [p] \mapsto f(p) \big)\ .\]
  Together with the orbit type stratification, the orbit space $X$ thus becomes a differentiable stratified space.    
\end{Example}

\begin{Example}
  Let $\calP$ be a poset with order relation $\leq$. By a $\calP$-\emph{stratified manifold}  we
  then understand a manifold $M$ endowed with a family $\left( M_{\leq P}\right)_{P\in \calP}$ of closed submanifolds such that
  for all $P,R \in \calP$ with $P < R$ the manifold $ M_{\leq P}$ is a closed submanifold of $M_{\leq R}$ with
  $\dim M_{\leq P} < \dim M_{\leq R}$. In case $M$ and all the manifolds $M_{\leq P}$ are Lie groups or vector spaces
  we call the pair $\left( M,\left( M_{\leq P}\right)_{P\in \calP}\right)$ a  $\calP$-\emph{stratified  Lie group} or
  $\calP$-\emph{stratified vector space}, respectively.
  Observe that a $\calP$-stratified manifold $M$ carries a natural stratification induced by the pieces  
  $M_R := M_{\leq R} \setminus M_{\leq P}$ where $P,R$ are elements of $\calP$ such that there exists no element
  $Q\in \calP$  with $P < Q <R$. 
  
  For the poset $\calP = \{ 1 < 2 < \ldots < n-1 <n \} $ examples of such $\calP$-\emph{stratified} manifold
  are the vector space $\R^n$ with the family of subvector spaces $\left(\R^k\right)_{1\leq k \leq n}$
  and the torus  $\mathbb{T}^n = \mathbb{S}^n$ with the family of subtori $\left(\mathbb{T}^k\right)_{1\leq k \leq n}$.
\end{Example}

Finally, we arrive at the regularity conditions (a) and (b) a stratified space can satisfy and which carry the name of their
inventor, Whitney.  Essentially, the Whintey conditions are both tameness requirements and entail that one can 
perform geometric and analytic constructions on stratified spaces fulfilling those regularity conditions. Condition (b)
is stronger than condition (a). Based upon Whitney's ground breaking work, Thom and Mather derived the two famous isotopy
lemmas from the regularity condition (b) and showed that as a consequence of the first isotopy lemma  (b)-regular
stratified spaces are locally trivial, see \cite{Thom,MatNTS}. 

\begin{Definition}
\label{Def-WhitneyReg}
  Let $(X,\calC^\infty(X))$ be a differentiable stratified space and assume without loss of generality that
  $X$ is a closed subspace of $\R^n$ and that $\calC^\infty(X) \cong \calC^\infty(\R^n) / \calJ_X$.
  Let $x$ be a point of $X$, $R$ the stratum in which $x$ lies and $S$ a stratum such that $R$ is incident to $S$ i.e.\ such that $R \subset \overline{S}$. Then one says that $X$ fulfills Whitney's condition
  \begin{enumerate}[(a)]
  \item if for every sequence $(y_k)_{k\in \N}$ in $S$ converging to $x$ and such that
    the sequence of tangent spaces  $(T_{y_k}S)_{k\in \N}$ converges in the Gra{\ss}mannian to a subspace $T \subset \R^n$
    the relation
    \[ T_xR \subset T \]
    holds true, 
  \item if   for every pair of sequences $(x_k)_{k\in \N} \subset R$ and $(y_k)_{k\in \N} \subset S$
    such that
    \begin{enumerate}[({b}1)]
    \item  both sequences converge to the same point $x\in R$,
    \item  the sequence of lines $(\ell_k)_{k\in \N}$, where $\ell_k$ is the line through
      $x_k$ and $y_k$, converges in projective space to a line $\ell \subset \R^n$,
    \item the sequence of tangent spaces  $(T_{y_k}S)_{k\in \N}$ converges in the Gra{\ss}mannian to a
      subspace $T \subset \R^n$,  
    \end{enumerate}
    the relation
    \[ \ell \subset T \]
    holds true. 
  \end{enumerate}
  Instead of saying that $X$ satisfies Whitney's condition (a) or (b), one calls $X$ 
  (a)-regular or (b)-regular, respectively.
\end{Definition}

\begin{Example}
  Every real algebraic or real analytic variety carries a unique minimal (b)-regular stratification. 
\end{Example}

\begin{Example}
\label{ex:hilbert.map-whitney-stratification}  
  Assume that $G$ is a compact Lie group which acts by linear transformations from the right on $\R^d$.
  Let $X =\R^d /G$ be the orbit space with the quotient topology.
  Denote by $\C[x_1,\ldots,x_d]^G $ the algebra of $G$-invariant polynomials on $\R^d$ and let
  $p_1,\ldots,p_k :\R^d \to \C $ be a Hilbert basis or in other words a generating system for the
  algebra $\C[x_1,\ldots,x_d]^G$. Then the Hilbert map
  \[
    p : \R^d \to \C^k , \quad x \mapsto
    \begin{pmatrix}
      p_1(x) \\ \vdots \\ p_k(x)
    \end{pmatrix}
  \]
  induces a chart $\overline{p} :X \to \C^k$ for the differentiable stratified space
  $\big(X, \calC^\infty (\R^d)^G \big)$. Moreover, the orbit type stratification of
  $\overline{p} (X) \cong X$ fulfills Whitney's condition (b) and actually coincides with the minimal (b)
  regular stratification of the orbit space; see \cite{BierstoneOrbitspaces} and \cite[Sec.\ 4.4]{PflaumBook}. 
\end{Example}

%


\section{Stratified Fiber Bundles}
\label{Stratif-Bunds}


Here we fix a useful notion of stratified fiber bundles, in particular stratified vector bundles, that
suffices for the finite dimensional case. Our definition ensures that there is a useful stratified
tangent bundle on the total space of a stratified fiber bundle.

\begin{Definition}

  Let $X$ be a Whitney (a)-regular stratified subspace of some $\RR^n$, denote by $\calS$ its stratification and
  by $\leq$ the incidence relation of strata.
  Assume that $E$ is a differentiable space and $\pi:E\to X$ a surjective differentiable map.
  Let $F$ be an $\calS$-stratified manifold and $G$ an $\calS$-stratified Lie group such that
  in both cases the induced stratifications are (a)-regular. 
  We then call $\pi$ a \emph{differentiable stratified fiber bundle} with fiber $F$ and structure group $G$ if the
  following is true:

\begin{enumerate}[{\rm (SB1)}]
\item For each $S\in \calS$, denote by $\pi_S:E_S\to S$ the pullback of $E$ over $S$. Then each
  $\pi_S$ is a differentiable locally trivial fiber bundle over $S$ with standard fiber manifold $F_{\leq S}$
  and structure group $G_{\leq S}$.
\item The decomposition $\{E_S \mid S \in \cal S\}$ is an (a)-regular stratification of $E$.
\item The map $\pi:E\to X$ is a topological quotient map.
\end{enumerate}

In particular, when $F$ is an $\calS$-stratified vector space with each fiber $F_{\leq S}$ a vector space,
and each $G_{\leq S}$ a closed subgroup of $\mathrm{Gl}(F_S)$, we will say that $E\to X$ is a \emph{differentiable stratified vector bundle}.

\end{Definition}

\begin{Example}
In \cite{PflaumBook}, the stratified tangent bundle $TX$ of a stratified sapace $X$ is defined by taking the union of all the classical tangent bundles $TS$ for the strata $S$.
When $X$ is endowed with a smooth structure compatible with the stratification, $TX$ is topologized so that singular charts give local homeomorphisms to the appropriate union of subbundles of the tangent bundle of euclidean space.
Then, in Theorem 2.1.2, the observation has been made that this gives a differentiable stratified vector bundle in the above sense if and only if Whitney's condition (a) holds for the stratified space $X$.
\end{Example}

\begin{Example}
Let $E\to X$ be a stratified fiber bundle.
Then the requirement above that $E$ be at least Whitney (a) stratified ensures that $E$ has a stratified tangent bundle $TE\to E$ which then is itself a stratified vector bundle so that the definition is coherent.
\end{Example}

It is useful to compare to examples of bundles that are stratawise fiber bundles but which fail to be stratified fiber bundles in the sense we give here.

\begin{Example}
Let $X=\{ (x,y,z) \in \bbR^3 \mid x^2 = y^2 |z| \, \& \, \mathrm{sgn}(x)=\mathrm{sgn}(yz) \}$.
This space is similar to the Whitney umbrella, i.e. the (real) algebraic variety  $x^2=y^2z$, with one panel of the umbrella canopy folded down.
Decompose $X$ into the $z$-axis $X_1$ and the remainder of the space $X_2=X\setminus X_1$.
Now consider the union $TX = TX_1 \cup TX_2$, where the union is taken in $T\bbR^3$ and where $TX$ has the subspace topology.
The restriction of the tangent bundle projection $T\bbR^3 \to \bbR^3$ to $TX\to X$ gives a bundle over $X$.
Moreover, on each stratum this is the usual tangent bundle $TX_i \to X_i$, which is a fiber bundle, in particular a vector bundle.
However, $TX\to X$ is not a quotient map.
	To see this, note that any nonzero tangent vector $v$ at the origin is in the inverse image of the closure of $X_2$ under the projection map, however, this vector is not in the closure of the inverse image of $X_2$ as this decomposition fails the Whitney (a)-regularity condition.
      \end{Example}
      
We remark here that other authors consider various notions of stratified bundles.
Note that in \cite{RossStratVectBunds} and \cite{ScarlettSmoothStratVectBunds} smooth structures are explicitly considered as part of the underlying structure, with the latter in particular defining a \emph{smooth stratified vector bundle} as a topological quotient map of stratawise smooth vector bundles with the additional requirement that there exist singular charts exhibiting the bundle as locally glued together from pullbacks of smooth vector bundles on open subsets of $\RR^n$ and with transition maps given locally by restrictions of smooth vector bundle isomorphisms.
The definition we give here is thus stronger than that of \cite{RossStratVectBunds}, which can include stratified tangent bundles over non Whitney (a) spaces, and is weaker than that of $\cite{ScarlettSmoothStratVectBunds}$, though every example we consider will satisfy the stronger requirements as well.
A definition is given in \cite{BauesFerrario} for stratifications in the form of filtrations by CW-pairs, and with fibers specified by families picked out via small functors whose image lies in locally compact spaces.
In the case of spaces stratified into finite dimensional smooth manifolds considered here, this should be a special case.

Now we consider maps between differentiable stratified fiber bundles.

\begin{Definition}
Let $E\to X$ and $E'\to X$ be two stratified fiber bundles over a stratified space $X$.
Then a bundle morphism $f:E\to E'$ is a \emph{morphism of stratified fiber bundles} if the following holds:

\begin{enumerate}[{\rm (SB1)}]
\setcounter{enumi}{3}
\item On the total space of each smooth fiber bundle over $S\in \calS$, the map $f_S: E_S\to E'_S$ given by
  restriction is a morphism of smooth fiber bundles.
\item In local coordinates, $f$ extends to a smooth function.
\end{enumerate}

In particular, if $f$ is an injection we will say that $E$ is a \emph{stratified subbundle} of $E'$.

\end{Definition}

Since here we ask that each $E_S\to S$ be a smooth fiber bundle in the classical sense, that is to say $E_S$ does not come
equipped with a further stratification, asking that a map commute with the bundle projections $E_S\to S$ suffices to ensure
that the map is also a stratified map, hence the codomain of the $f_S$ does indeed land in $E'_S$.

\begin{Example}
Let $\pi:E\to X$ be a stratified fiber bundle where $X$ is a subset of $\RR^n$.
Then $T\pi:TE\to TX$, given by taking the total derivative stratawise, is a morphism, which we call the
\emph{stratified derivative} or \emph{stratified pushforward} associated to $\pi$.
\end{Example}

We wish to consider conditions that allow for taking kernels stratawise for $T\pi$ to obtain a subbundle and an inclusion
morphism. To that end we introduce the notion of Artin gluing data for fiber bundles over stratified spaces.
Let $(X,\calS)$ be a differentiable stratified space and let $E\to X$ be a map which is stratawise a smooth fiber bundle.
Let $\Gamma E_S$ denote the sheaf of sections of $E\to X$ pulled back to $S$ along the inclusion $S\hookrightarrow X$.
For each $R<S$ one has the bundle $f_{RS}\Gamma E_S$ on $R$ of \emph{fringe germs} whose elements can be represented by
taking open subsets $U$ of $X$ around a point $x\in R$ and considering sections of $E\to X$ over $U\cap S$ with two
sections equivalent when they agree on an open subset of their intersection.
Note that these sections are not required to be defined at $x$.

\begin{Definition}
	The \emph{induced Artin gluing data} for the sheaves $\Gamma E_S$, $S\in \calS$ consist  of  the maps 
	\[
		g_{RS}: \Gamma E_R \to f_{RS}\Gamma E_S  \quad \text{for } R,S \in \calS \text{ with } R< S 
	\]
	given by mapping a germ of a section $\sigma:U\to E$ at some $x\in R$ to the fringe germ at $x$ of $\sigma|_{U\cap S}$.
	Note that the maps $g_{RS}$ with $R< S$ are compatible in the sense that when $Q<R<S$, given a section $\sigma$ on $U$ with germ
        $[\sigma]_x$ at $x\in Q$ there is a natural inclusion map $f_{_{QS}} \Gamma E_S\to f_{_{QR}}f_{RS}\Gamma E_S$ and these maps commute up to these inclusions.

	The \emph{induced Artin gluing data} for the bundles $E_S$ consist of the above $g_{RS}$ with $R<S$ 
	together with the evaluation maps 
	\[
		\epsilon_S: \Gamma E_S \to E_S \quad \text{for } S \in \calS \ ,
	\]
	given by taking a germ of some function $\sigma:U\to E$ at $x$ to $\sigma(x)$.

	We say that Artin gluing data are \emph{injective} if for each stratum $R$ the maps $g_{RS}$ with $R<S$ are
        jointly injective and the maps $\epsilon_S$, $S\in \calS$ are  surjective.

        Hereby, when we say that for fixed $R$ the family $\left(g_{RS}\right)_{R<S}$ is \emph{jointly injective} we mean
        that if $x,y \in R $ and $g_{RS} (x) = g_{RS} (y)$ for all strata $S$ with $R<S$ then $x=y$. 
	
\end{Definition}

Artin gluing data are then the coherent specification of which fringe germs at $x$ which arise from actual germs at $x$, as well as how the sections give elements of the fiber bundles.
Given abstract Artin gluing data between arbitrary sheaves (or general bundles) on the strata $S\in \calS$ one can construct a sheaf (or bundle) on the space $X$ whose
restriction to each stratum $S$ is the given sheaf and whose induced Artin gluing data are equivalent to the given data, as in \cite{JordanThesis}.
Note that if there are two strata $R$ and $S$ which are not adjacent, then $f_{RS}$ is the trivial functor taking every
sheaf on $S$ to the terminal sheaf. In such a case, the map $g_{RS}$ is certainly not injective.
Moreover, it is possible that two sections might agree on a lower stratum but disagree on a higher one. 
However, in our case where we consider only adjacent strata and sheaves on the strata given by restrictions of sections over
$X$, we can show that this does not occur.

\ble
	Let $X$ be a stratified space and $T \to X$ a surjective continuous map.
	Then the maps $g_{RS}$ for the induced Artin gluing data on the sheaves are jointly injective.
\ele

%

\bpf
	Let $\sigma,\sigma':U\to T$ be two sections with U a neighborhood $x\in R$, where $U$ intersects only finitely many strata.
	Denote the germs of these sections at $x$ by $\sigma_x$ and $\sigma'_x$ respectively.
	Suppose that for all $S>R$, $g_{RS}(\sigma_x) = g_{RS}(\sigma'_x)$ i.e.\ that the sections give the same fringe germs at $x$ for all higher strata.
	This means exactly that for each $S$ there is a neighborhood $W_S\subset U$ of x for which $\sigma\mid_{W_S\cap S} = \sigma'\mid_{W_S\cap S}$.
	Let $W=\int W_S$. Since there are only finitely many $S$ such that $S>R$, $W$ is open.
	Then $\sigma$ and $\sigma'$ coincide on all $W\cap S$.
	By the frontier condition, $W\cap R\subset \overline{W\cap S}$ for all $S > R$, hence $\sigma$ and $\sigma'$ agree on all of $W$.
	Thus, the germs of $\sigma$ and $\sigma'$ at $x$ are equal.
\epf

For a stratified but not Whitney (a)-stratified subset $X$ of $\RR^n$ and for the union of the tangent bundles of the strata, the maps $g_{RS}$ are all injective but the maps $\epsilon_S$ are not surjective.
Hence, there is a sense in which the fibers over lower strata can not be included into the fibers of the higher strata.

\ble
	Let $E\to X$ be a stratified fiber bundle.
	Let $V\hookrightarrow E$ be a subbundle over $X$ such that stratawise $V_S\hookrightarrow E_S$ is a
        differentiable fiber subbundle.	In particular, $V\subset E$ carries the subspace topology.
	Then the following are equivalent:

	\begin{enumerate}[(1)]
		\item $V\to X$ is a stratified subbundle of $E \to X$.
		\item The induced map $V\to X$ is a topological quotient map and $\{V_S\}$ is a stratification of $V$.
		\item For every pair of strata $R < S$, every element $v\in V_R$ is a limit of a sequence in $V_S$.
		\item The induced Artin gluing data are injective.
	\end{enumerate}

\ele

\bpf
	The second condition is part of the definition we gave for differentiable fiber bundles here, and $T\to X$ is already assumed to be a differentiable stratified fiber bundle with subspace $V$, hence the first two statements are equivalent.
	The third condition is just an equivalent definition of the quotient topology.
	
	Consider now the last condition.
	Note that in order for the maps $\Gamma E_S\to E_S$ to be surjective, it must be that every element of $E_S$ is in the image of some local section, and this implies the third condition.
	So, it remains to show that injectivity of the Artin gluing data follows from the other conditions.
	Suppose now that $V\to X$ is a differentiable stratified subbundle of $\to X$, with fiber $F$.
	Note that by the prior lemma, we only need to show that the evaluation maps are surjective.
	That is, we need that every element $e\in E_R$ over a point $x\in R$ is in the image of some continuous section around $x$.
	Let $U$ be a neighborhood of some $x\in R$ such that each $U_Q = U\cap Q$ for $Q\geq R$ trivializes the fiber bundle, via some isomorphism $\phi_Q$ with $U_Q \times F_{\leq Q}$.
	Note that by local finiteness, such a $U$ is guaranteed to exist.
	Let the section $\sigma$ be given locally on $U$ by the constant function with value $\phi_Q (e)$.
	Then by construction, $\epsilon_R(\sigma_x)=e$, where $\sigma_x$ is the germ of sigma at $x$.
\epf

Hence, in order for the kernels of the $T\pi$ over the various strata to form a subbundle we must have that every element
in the kernel of $T\pi_R$ is the limit of a sequence in the kernel of $T\pi_S$ whenever $R<S$.
That is, the map $T\pi$ must satisfy Thom's functional regularity condition $(a_f)$.
We show this now.

\ble
Let $\pi:E\to X$ be a stratified fiber bundle and $T\pi:TE\to TX$ the stratified pushforward.
Let $VE_S\to TE_S$ be the kernel of $T\pi_S:TE_S\to TS$ for each stratum $S$ endowed with the subspace topology.
Let $v$ in $VE_R$ be given, i.e.\ assume that $T\pi_R(v) = 0$.
Then for any stratum $S$ adjacent to $R$, there is a sequence $(v_k)_{k\in \N}$ in $VE_S$ converging to $v$.

\ele

\bpf
Let $x=\pi(e)$ be in the stratum $R$ and let $S>R$ be a higher stratum adjacent to $R$.
Let $U$ be a neighborhood of $x$ such that $U_R=U\cap R$ and $U_S=U\cap S$, respesctively, trivialize $E_R$ and $E_S$.
That is, the pullbacks of $E$ over these open sets are isomorphic to the trivial bundles $U_R\times F_{\leq R} \to U_R$
and $U_S\times F_{\leq S} \to U_S$, respectively, and the pullback of $TE$ over these are
	\[
		TU_R \oplus TF_{\leq R} \to U_R \times F_{\leq R}
	\]

	and

	\[
		TU_S \oplus TF_{\leq S} \to U_S\times F_{\leq S} \ ,
	\]
        respectively. The kernels $VE_R$ and $VE_S$ are isomorphic to $TF_{\leq R}$ and $TF_{\leq S}$, respectively.

	Under the trivialization, $e$ is represented by a pair $(x,f)$ and $v$ is represented by a differentiable path $\gamma$ in $U_R\times F_{\leq R}$ centered at $(x,f)$.
	By the frontier condition on the decomposition of $E$, there is some sequence $(f_k)_{k\in \N}$ in $E_S$ with
        limit $f$, and without loss of generality we may assume that it is contained in $U_S$.

	We may assume, after shrinking $U$ if necessary, that the tangent bundle $TE$ is also trivialized over
        $\pi^{-1}(U_R)$ and $\pi^{-1}(U_S)$.
	Hence, $v$ is represented by some $(f,\nu) \in VE_R\subset F_{\leq R} \times \bbR^d$, where $d$ is the dimension of $F_{\leq R}$.
	$TF_{\leq R}$ embeds in $TF_{\leq S}$ in a way compatible with the trivializations, hence the sequence
        $(f_k,\nu)_{k\in \N}$ with limit $(f,\nu)$ represents a sequence $(v_k)_{k\in \N}$ with limit $v$.

\epf

%


\section{Stratified Connections}
\label{Stratif-Conns}


We first recall the definition of an Ehresmann connection on a smooth
fiber bundle $\pi:E\to M$ and then consider a stratified generalization.
Assuming that $\pi:E\to M$ is a smooth map between finite dimensional manifolds one has the tangent map $T\pi : TE\to TM$.
The map $T\pi$ is a smooth vector bundle morphism, hence has a kernel $\ker T\pi$ which by definition coincides with
the verticle bundle $VE \hookrightarrow TE$.
In this context a connection is simply a smooth choice of a splitting vector bundle $HE\to TE$, i.e.

\begin{Definition}
An \emph{Ehresmann connection} on $\pi$ is a splitting of the short exact sequence
\[
  0 \longrightarrow VE \longrightarrow TE \overset{T\pi}{\longrightarrow} TM \to 0 
\]
in the category of smooth vector bundles.

That is, for each $x\in M$, $e\in E_x$ one has a vector subspace $HE_e$ of $T_eE$ isomorphic to $T_xM$ via $T\pi$ and
such that $TE_e=VE_e\oplus HE_e$. Taken together the $H_eE$ form a smooth vector subbundle $HE$ of $TE$ over $E$
which is complementary to $VE$.
\end{Definition}

If $E\to M$ is locally trivial with fiber $F$, then $TE\to TM$ is locally trivial as well, and we have also that $V_eE$ is
isomorphic to $TE_x$.
By considering the associated $G$-principal bundle where $G$ is the structure group of $E\to M$, the above is equivalent to
choosing a smooth $1$-form $\omega$ on $M$ valued in the Lie algebra $\g$ of $G$. 
Then $HE$ is the (image of the) kernel of $\omega$.

Now suppose that $X$ is Whitney (a) stratified space with stratification $\calS$, and that $E\to X$ is a stratified
fiber bundle over $\calS$.

Since $VE$ is a stratified subbundle, we have another exact sequence
\[
	0 \longrightarrow VE \longrightarrow TE\longrightarrow TX\longrightarrow 0 \ ,
\]
this time in the category of stratified vector bundles over $\calS$.

A stratified connection should be again a splitting of this sequence.
\begin{Definition}
Let $\pi:E\to X$ be a stratified fiber bundle.
A \emph{stratified Ehresmann connection} on $\pi$ is a splitting of 
\[
  0 \longrightarrow VE \longrightarrow TE \overset{T\pi}{\longrightarrow} TX \longrightarrow 0 
\]
in the category of stratified vector bundles over $\calS$.
That is, a stratified Ehresmann connection is  a stratified vector subbundle $HE\to TE$ such that $HE\cong TX$ and
$TE=VE\oplus TX$.
\end{Definition}

Note that the splitting must respect the stratification in order that $HE$ be a stratified vector subbundle.
In terms of the connection $1$-forms $\omega_R$ and $\omega_S$ when $R,S$ are strata of $X$ with $R<S$, this means that
every element in the kernel of $\omega_R$ must be a limit of elements of $\omega_S$.

Thus, we have the following

\begin{Lemma}
\label{StratConn-lem}
	Let $\pi:E\to X$ be a stratified fiber bundle over $X$ and $VE\to TE$ be the kernel of $T\pi:TE\to TX$.
	Assume that for each stratum $S$ of $X$ the bundle  $HE_S$ is a connection on $E_S\to S$.
	Then the union \[ HE = \bigcup_{S\in \calS}\bigcup_{e\in E_S} H_eE_S\] constitutes a stratified connection on
        $E\to X$ if and only if every $h\in HE_R$ with $R <S$, $R,S\in \calS$ is a limit of a sequence
        $(t_k)_{k\in\N}$ in $HE_S$.
\end{Lemma}

This makes the stratawise splitting of the vector bundle $TE=VE\oplus HE$ coherent in the sense that one understands when
a decomposition $w_e=v_e+h_e$ of some $w_e\in T_e E_R=V_e E_R \oplus H_e E_R$ is given as the limit of a
sequence $w_k=v_k+h_k$, $k\in \N$ of elements of the bundle over a higher stratum $TE_S=VE_S\oplus HE_S$.
Note that unlike in the case of the $VE_S$ we do not automatically have a globally defined $1$-form $\omega$ to take the
kernel as each $\omega_S$ has codomain in the Lie algebra of $G_S$.
However, Artin gluing data induce Lie group embeddings $G_R\to G_S$ for an (a)-regular stratified fiber bundle. In this case
we therefore have an embedding of Lie algebras as well and this allows us to define $\omega$ globally with codomain the Lie
algebra of the Lie group colimit of the $G_R\to G_S$.
Note that even if the $G_S$ are all finite dimensional Lie groups but the stratification is unbounded then the directed
colimit might not only be infinite dimensional but even can fail to be, say, Fr\'echet even though all the $G_S$ are so.
Whenever the stratification has a maximal element, though, e.g.\ when a dense open top stratum exists, then the directed colimit is
simply the structure group for the fiber bundle over the top stratum, hence is within the same category of Lie groups
respectively Lie algebras.


\end{document}